\documentclass[12pt]{article}

\usepackage[utf8]{inputenc}
\usepackage[T1]{fontenc}
\usepackage{eucal}
\usepackage{microtype}
\usepackage{csquotes}
\usepackage{mathtools}
\usepackage{amssymb,amsfonts,stmaryrd}
\usepackage{amsthm}
\usepackage{tikz-cd}
\usepackage[citestyle=numeric-comp,sorting=nyt,sortcites, backend=biber, giveninits=true, maxnames=99]{biblatex}
\usepackage{enumitem}
\setlist[enumerate,1]{label={\roman*)}}
\usepackage[unicode=true, pdfusetitle, colorlinks=true, allcolors=blue, bookmarks=false]{hyperref}
\usepackage{authblk} %affiliations
\usepackage{graphicx,subcaption}
\newbox{\myorcidaffilbox}
\sbox{\myorcidaffilbox}{\large\includegraphics[height=1.7ex]{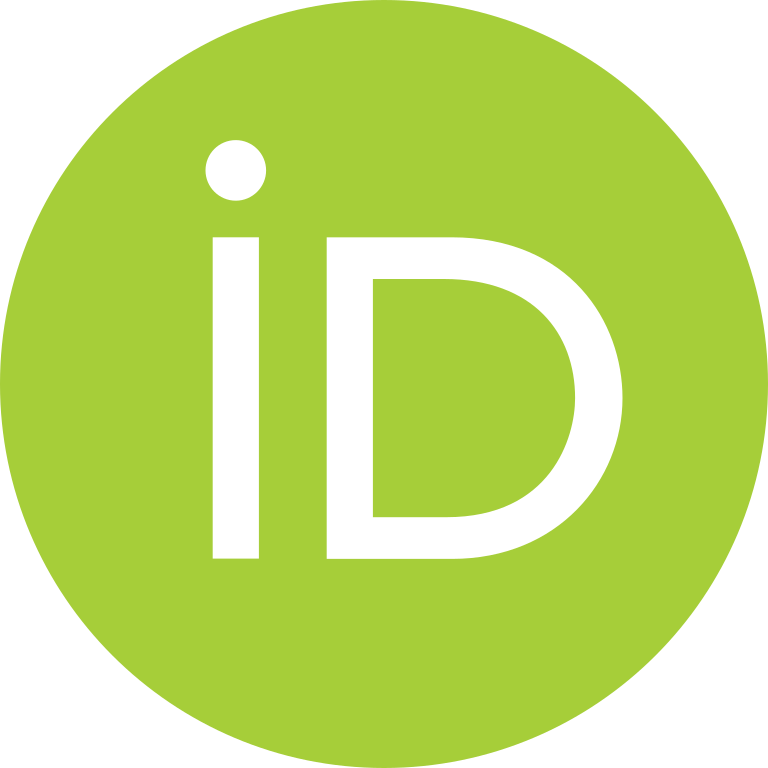}}
\newcommand{\orcid}[1]{%
  \href{https://orcid.org/#1}{\usebox{\myorcidaffilbox}}}
\usepackage[left=2.5cm,right=2.5cm,top=3cm,bottom=3.5cm]{geometry}
\theoremstyle{plain}
\newtheorem{theorem}{Theorem}
\newtheorem*{theorem*}{Theorem}
\newtheorem{lemma}[theorem]{Lemma}
\newtheorem*{lemma*}{Lemma}
\newtheorem{proposition}[theorem]{Proposition}
\newtheorem*{proposition*}{Proposition}
\newtheorem{corollary}[theorem]{Corollary}  
\newtheorem*{corollary*}{Corollary}
\theoremstyle{definition}
\newtheorem{definition}{Definition}
\newtheorem*{definition*}{Definition}
\newtheorem{example}{Example}
\newtheorem*{example*}{Example}
\theoremstyle{remark}
\newtheorem{remark}{Remark}
\newtheorem*{remark*}{Remark}
\newtheorem{conjecture}{Conjecture}
\newtheorem*{conjecture*}{Conjecture}

\newtheorem*{problem*}{Problem}

\newcommand*{\RR}{\mathbb{R}}

\newcommand*{\dd}{\mathrm{d}}

\newcommand*{\contr}[1]{\iota_{#1}}
\newcommand*{\liedv}[1]{\mathcal{L}_{#1}}

\DeclareMathOperator{\Leg}{Leg}

\newcommand{\Legp}{\mathbb{F}^{f+} L_{d}}
\newcommand{\Legm}{\mathbb{F}^{f-} L_{d}}

\newcommand{\norm}[1]{\left\lvert\left\lvert #1 \right\rvert\right\rvert}

\bibliography{biblio}

\AtEveryBibitem{
  \ifentrytype{online}{}{% Remove url except for @online
    \ifentrytype{thesis}{}{
        \clearfield{url}
      }
  }
}

\DeclareSourcemap{
  \maps[datatype=bibtex]{
    \map[overwrite=true]{
      \step[fieldset=urldate, null]
      \step[fieldset=language, null]
      \step[fieldset=address, null]
     \step[fieldset=pagetotal, null]
    }
  }
}

\title{Discrete Hamilton-Jacobi theory for systems with external forces}
\author[1,2]{\orcid{0000-0002-8028-2348} Manuel de León }
\author[1]{\orcid{0000-0002-2368-5853} Manuel Lainz}
\author[1]{\orcid{0000-0002-9620-9647} Asier López-Gordón\thanks{Author to whom correspondence should be addressed: \href{mailto:asier.lopez.gordon@csic.es}{asier.lopez@icmat.es}}
}
\affil[1]{Instituto de Ciencias Matemáticas (CSIC-UAM-UC3M-UCM) \protect\\
Calle Nicolás Cabrera, 13-15, Campus Cantoblanco, UAM, 28049 Madrid, Spain}
\affil[2]{Real Academia de Ciencias Exactas, Físicas y Naturales\protect\\
Calle Valverde, 22, 28004, Madrid, Spain }
\date{\today}

\mathtoolsset{showonlyrefs}

\allowdisplaybreaks %allow page breaks of aligned equations
\begin{document}

\maketitle

\begin{abstract}
\noindent
This paper is devoted to discrete mechanical systems subject to external forces. We introduce a discrete version of systems with Rayleigh-type forces, obtain the equations of motion and characterize the equivalence for these systems. Additionally, we obtain a Noether's theorem and other theorem characterizing the Lie subalgebra of symmetries of a forced discrete Lagrangian system. Moreover, we develop a Hamilton-Jacobi theory for forced discrete Hamiltonian systems. These results are useful for the construction of so-called variational integrators, which, as we illustrate with some examples, are remarkably superior to the usual numerical integrators such as the Runge-Kutta method.

\bigskip

\end{abstract}

% \tableofcontents

\tableofcontents

% \todo[inline, color=green!40, caption={}]{
%     \begin{itemize}
%       % \item Discrete symmetries
%       % \item $(L_d, \mathcal{R}_d)\cong (L_d', \mathcal{R}_d')$
%       % \item $\exists\ \mathcal{R}_d$
%       % \item Simulation 
%       % \item Introduction
%       \item Examples
%       % \item Abstract
%       \item Conclusions
%     \end{itemize}
%   }

\section{Introduction}

In previous articles we have initiated a study of mechanical systems with external forces from a geometrical point of view, with the intention of completing some gaps on the subject that do not seem to be covered in the literature. In particular, we have paid special attention to Rayleigh-type forces, which classically correspond to dissipative systems whose dissipation rate is linear with respect to velocities or, equivalently, momenta.

A modern Lagrangian formulation of a system with external forces (systems referred to there as non-conservative) can be found in Godbillon's book \cite{godbillon_geometrie_1969}. However, as far as we know, the references found in the literature concerning the discretization of these systems are reduced to the study in the paper by Marsden and West \cite{marsden_discrete_2001} (whose results on the subject we have improved in this paper), an article by Lew, Marsden, Ortiz and West \cite{lew_variational_2004}, and the studies by Sato and Martín de Diego \cite{sato_martin_de_almagro_discrete_2020,diego_variational_2018,diego_variational_2020} 
% (with a different approach to ours).
whose approach, different from ours, consists on doubling the degrees of freedom and reducing the forced system to a double-dimensional conservative system (see also Ref.~\cite{galley_principle_2014}).

The symplectic description for the Hamiltonian formulation, and the almost tangent geometry in the case of Lagrangian systems, has allowed us to address in a systematic way a series of questions almost impossible otherwise; for example, the study of infinitesimal symmetries and the corresponding conserved quantities, or the description of a Hamilton-Jacobi theory that allows us to study the complete solutions and integrability of the system.

In this paper, we study the discrete version of a mechanical system subject to an external force. Such a system is determined by a discrete Lagrangian $L_d : Q \times Q \to \RR$ defined in the product of the configuration space $Q$ by itself, and an external force which is a 1-form $f_d$ in $Q \times Q$ and therefore composed of two 1-forms in $Q$, the left $f_d^{-1}$ and the right $f_d^{+}$, which obey the discrete Lagrange-d'Alembert principle, giving the equations
\begin{equation}
  D_{2} L_{d}\left(q_{0}, q_{1}\right)+D_{1} L_{d}\left(q_{1}, q_{2}\right)+f_{d}^{+}\left(q_{0}, q_{1}\right)+f_{d}^{-}\left(q_{1}, q_{2}\right)=0. 
\end{equation}
Using the two Legendre transformations, left and right, we can define the left and right Hamiltonians, the corresponding discrete Hamilton equations, and the discrete Hamiltonian flow. We also consider the case when the external forces are of the Rayleigh type. 

This formulation has allowed us to study a number of issues such as the extension of Noether's theorem in the presence of symmetries, the notion of discrete Rayleigh potential, the relations between the continuous system and its discrete approximation, or to establish a Hamilton-Jacobi equation. One of the consequences of our work is the possibility of defining numerical integrators that behave in an excellent way as we show in several examples.

The paper is structured as follows.  In Section \ref{section_continuous} we review continuous Hamiltonian and Lagrangian mechanics in the symplectic formulation. We also recall the relation between fibred morphisms and semibasic 1-forms, as well as the Lagrange-d'Alembert principle. In Section \ref{section_discrete} we recall discrete Lagrangian mechanics with external forces, derived from the discrete Lagrange-d'Alembert principle. As a novelty, we introduce the notion of discrete Rayleigh forces. When the discrete force is of Rayleigh-type, the Euler-Lagrange equations and the Legendre transforms can be expressed in terms of two modified discrete Lagrangians $L_d^\pm$. We study the equivalence of discrete Rayleigh systems. Moreover, we present a discrete Noether's theorem, together with other theorem characterizing the subalgebra of symmetries of the discrete Lagrangian which also leaves the discrete force invariant. These two theorems extend a result previously found by Marsden and West \cite{marsden_discrete_2001}. In Section \ref{section_HJ}, we construct a Hamilton-Jacobi theory for discrete systems with external forces. We make use of the discrete flow approach proposed by de León and Sardón \cite{de_leon_geometry_2018}, extending their results for systems with external forces. We obtain a theorem relating the solutions of the Hamilton-Jacobi equations with the solutions of the Hamilton equations, extending the result by Ohsawa, Bloch and Leok \cite{ohsawa_discrete_2011} for discrete systems with external forces. Section \ref{section_conclusion} presents some conclusions and open problems.

\textbf{Notation.}
 Throughout this paper, let $Q$ be an $n$-dimensional differentiable manifold, which represents the configuration space of a dynamical system. Let $T_qQ$ and $T_q^*Q=(T_qQ)^*$ denote the tangent and cotangent spaces of $Q$ at the point $q\in Q$.
  % Its dual vector space, $T_q^*Q$, is called the cotangent space of $Q$ at $q$.
   Let $\tau_Q:TQ\to Q$ and $\pi_Q:T^*Q\to Q$ be its tangent bundle and its cotangent bundle, respectively; namely, $TQ=\cup_{q\in Q} T_qQ$ and $T^*Q=\cup_{q\in Q} T_q^*Q$, with the canonical projections $\tau_Q:(q^i, \dot q^i)\mapsto (q^i)$ and $\pi_Q:(q^i, p_i)\mapsto (q^i)$.
   Unless otherwise stated, sum over paired covariant and contravariant indices is understood.
 % \todo[inline]{Introduce $T_qQ, T_q^*Q$}

 Given a smooth manifold $M$, we denote by $\Omega^p(M)$ the real vector space (and $C^\infty$-module) of $p$-forms on $M$. We denote by $\mathfrak X(M)$ the real vector space (and $C^\infty$-module) of vector fields on $M$. For each $\alpha \in \Omega^p(M)$ and each $X\in \mathfrak{X}(M)$, $\contr{X} \alpha \in \Omega^{p-1}(M)$ denotes the interior product of $\alpha$ by $X$, and $\liedv{X} \alpha$ denotes the Lie derivative of $\alpha$ with respect to $X$.

 Given a function or map $f$ with $k$ arguments, $D_i f$ denotes the derivative of $f$ with respect to its $i$-th argument. For instance, given a discrete Lagrangian function $L_d:Q\times Q\to \RR$, $D_1 L_D(q_i, q_j)$ denotes the derivative of $L_d$ with respect to $q_i$.

 % \todo[inline, caption={}, color=green!40]{
 % \begin{itemize}
 % \item $D_i$ para las derivadas
 % \item Explicar $TQ$ y $T^*Q$
 
 % \end{itemize}}

\section{Continuous mechanics with external forces: a review}\label{section_continuous}

\subsection{Semibasic forms and fibred morphims}\label{section_morphisms}

Consider a fibre bundle $\pi: E\to M$. Let us recall \cite{leon_methods_1989,abraham_foundations_2008,godbillon_geometrie_1969} that a 1-form $\beta$ on $E$ is called \emph{semibasic} if 
\begin{equation}
    \beta(Z)=0
\end{equation}
for all vertical vector fields $Z$ on $E$. If $(x^i,y^a)$ are fibred (bundle) coordinates, then the vertical vector fields are locally generated by $\{\partial/\partial y^a\}$. So $\beta$ is a semibasic 1-form if it is locally written as
\begin{equation}
  \beta=\beta_i(x,y) \dd x^i.
\end{equation}

Let us now particularize this definition for the cases of tangent and cotangent bundles.
 % Hereinafter, let $Q$ be an $n$-dimensional differentiable manifold. Let $\tau_Q:TQ\to Q$ and $\pi_Q:T^*Q\to Q$ be its tangent bundle and its cotangent bundle, respectively.
Consider a morphism of fibre bundles, i.e., a map $D:TQ\to T^*Q$ such that the following diagram commutes:
\begin{center}
\begin{tikzcd}
 TQ \arrow[rr, "D"] \arrow[rd, "\tau_q"'] &   & T^*Q \arrow[ld, "\pi_Q"] \\
                                                          & Q &                         
\end{tikzcd}
\end{center}
One can define a corresponding semibasic 1-form \cite{godbillon_geometrie_1969,leon_methods_1989} $\beta_D$ on $TQ$ by 
\begin{equation}
    \beta_D(v_q)(u_{v_q})=\left\langle D(v_q),T\tau_Q(u_{v_q})\right\rangle,
\end{equation}
where $v_q\in T_qQ,\ u_{v_q}\in T_{v_q}(TQ)$.

Suppose that locally $D$ is given by
\begin{equation}
    D(q^i,\dot q^i)=(q^i,D_i(q,\dot q)),
\end{equation}
% where $D_i\ (i=1,\ldots, n)$ are local functions on $TQ$, 
in other words, to each vector $\dot q^i \partial/\partial q^i\in T_qQ$, it assigns the covector $D_i(q, \dot q) \dd q^i\in T^*_qQ$. 
Then,
\begin{equation}
    \beta_D=D_i(q,\dot q)\dd q^i.
\end{equation}

Conversely, given a semibasic 1-form $\beta$ on $TQ$, one can define the following morphism of fibre bundles:
% \begin{center}
% \begin{tikzcd}
% D_\beta: TQ \arrow[rr] \arrow[rd, "\tau_q"'] &   & T^*Q \arrow[ld, "\pi_Q"] \\
%                                                           & Q &                         
% \end{tikzcd},
% \end{center}
\begin{equation}
\begin{aligned}
    &D_\beta: TQ\to T^*Q,\\
    &\left\langle D_\beta(v_q),w_q\right\rangle=\beta(v_q)(u_{w_q}),
\end{aligned}
\end{equation}
for every $v_q,w_q\in T_qQ,\ u_{w_q}\in T_{w_q}(TQ)$, with $T\tau_Q(u_{w_q})=w_q$. In local coordinates, if
\begin{equation}
\beta=\beta_i(q,\dot q)\dd q^i,
\end{equation}
then
\begin{equation}
    D_\beta(q^i,\dot q^i)=\left(q^i,\beta_i(q^i,\dot{q}^i)\right).
\end{equation}
Here $(q^i, \dot{q}^i)$ are bundle coordinates in $TQ$.

So there exists a one-to-one correspondence between semibasic 1-forms and fibred morphisms from $TQ$ to $T^*Q$.

\subsection{Lagrange-d'Alembert principle}
Consider a forced Lagrangian system $(L,f_L)$ on $TQ$, with  Lagrangian function $L:TQ\to \RR$ and external force $f_L$. Let us recall that $f_L$ is a fibre-preserving map $f_L:TQ\to T^*Q$ over the identity, locally given by
\begin{equation}
  f_{L}:(q, \dot{q}) \mapsto\left(q, f_{L}(q, \dot{q})\right).
\end{equation}
Here $(q,\dot{q})=(q^1,\ldots,q^n,\dot{q}^1, \ldots, \dot{q}^n)$ for the ease of notation.
As we have explained in Subsection \ref{section_morphisms}, it can be equivalently seen as a semibasic 1-form. The dynamics of the \emph{forced Lagrangian system} $(L,f_L)$ on $TQ$ are given by the \emph{Lagrange-d'Alembert principle} \cite{marsden_discrete_2001,lew_variational_2004} (see also \cite{lanczos_variational_1986})
\begin{equation}
\delta \int_{0}^{T} L(q(t), \dot{q}(t))\ \mathrm{d} t+\int_{0}^{T} f_{L}(q(t), \dot{q}(t)) \cdot \delta q(t) \mathrm{d} t=0,
\end{equation}
where $\delta$ denotes variations vanishing at the endpoints. This principle leads to the \emph{forced Euler-Lagrange equations} 
\begin{equation}
\frac{\partial L}{\partial q}(q, \dot{q})-\frac{\mathrm{d}}{\mathrm{d} t}\left(\frac{\partial L}{\partial \dot{q}}(q, \dot{q})\right)+f_{L}(q, \dot{q})=0.
\label{forced_continuous_EL}
\end{equation}
In the absence of external forces, the Lagrange-d'Alembert principle reduces to the well-known Hamilton principle,
\begin{equation}
  \delta \int_{0}^{T} L(q(t), \dot{q}(t))\ \dd t = 0,
\end{equation}
from where the usual Euler-Lagrange equations,
\begin{equation}
  \frac{\partial L}{\partial q}(q, \dot{q})-\frac{\mathrm{d}}{\mathrm{d} t}\left(\frac{\partial L}{\partial \dot{q}}(q, \dot{q})\right)=0,
\end{equation}
can be derived.

\subsection{Hamiltonian systems with external forces}
An external force is geometrically interpreted as a semibasic 1-form on $T^*Q$. A Hamiltonian system with external forces (so called \emph{forced Hamiltonian system}) $(H,\beta)$ is given by a Hamiltonian function $H:T^*Q\to \mathbb{R}$ and a semibasic 1-form $\beta$ on $T^*Q$. Let $\omega_Q=-\dd \theta_Q$ be the canonical symplectic form of $T^*Q$.
% Introduce $\omega_Q$
% Introduce manifold?
Locally these objects can be written as
\begin{equation}
\begin{aligned}
&\theta_Q=p_i\dd q^i,\\
&\omega_Q=\dd q^i\wedge \dd p_i,\\
&\beta=\beta_i(q,p) \dd q^i,\\
&H=H(q,p),
\end{aligned}
\end{equation}
% where $(q,p)=(q^1,\ldots,q^n,p_1,\ldots,p_n)$.
where $(q^i,p_i)$ are bundle coordinates in $T^*Q$. At each point $(q,p)\in T^*Q$, $\beta_i(q,p)\dd q^i\in T^*_qQ$ is a covector.

The dynamics of the system is given by the vector field $X_{H,\beta}$, defined by
\begin{equation}
\contr{X_{H,\beta}}\omega_Q=\dd H+\beta. \label{Hamiltonian_dynamics_eq}
\end{equation}
If $X_H$ is the Hamiltonian vector field for $H$, that is,
\begin{equation}
\contr{X_H}\omega_Q=\dd H, \label{Hamiltonian_VF}
\end{equation}
and $Z_\beta$ is the vector field defined by
\begin{equation}
\contr{Z_\beta}\omega_Q=\beta,
\end{equation}
then we have
\begin{equation}
X_{H,\beta}=X_H+Z_\beta.
\end{equation}
Locally, the above equations can be written as 
\begin{equation}
\begin{aligned}
&X_H=\frac{\partial  H} {\partial p_i }\frac{\partial  } {\partial  q^i}
-\frac{\partial H } {\partial  q^i} \frac{\partial  } {\partial  p_i}
 \label{Hamiltonian_VF_local},
\\
&\beta=\beta_i \dd q^i,\\
&Z_\beta=-\beta_i \frac{\partial  } {\partial p_i} ,\\
&X_{H,\beta}=\frac{\partial  H} {\partial p_i }\frac{\partial  } {\partial  q^i}-\left(\frac{\partial H } {\partial  q^i}+\beta_i\right)
\frac{\partial  } {\partial  p_i}.
\end{aligned}
\end{equation}
Then, a curve $(q^i(t), p_i (t))$ in $T^*Q$ is an integral curve of $X_{H,\beta}$ if and only if it satisfies the forced motion equations
\begin{equation}
\begin{aligned}
&\frac{\mathrm{d}q^i} {\mathrm{d}t}=\frac{\partial  H} {\partial p_i },\\
&\frac{\mathrm{d}p_i} {\mathrm{d}t}=-\left(\frac{\partial  H} {\partial q^i }+\beta_i\right).
\end{aligned}
\end{equation}
% Let us recall that the \emph{Poisson bracket} is the bilinear operation
% \begin{equation}\\
% \begin{aligned}
% \left\{\cdot,\cdot  \right\}: C^\infty(M,\RR)\times C^\infty(M,\RR) &\to C^\infty(M,\RR)\\
% \left\{f,g  \right\} &=\omega(X_f,X_g),
% \end{aligned}
% \end{equation}
% with $X_f, X_g$ the Hamiltonian vector fields associated to $f$ and $g$, respectively.

\subsection{Lagrangian systems with external forces}
% \todo[inline, color=red!40]{Revise sign criteria! $f_L=-\alpha= - \dd_S \mathcal{R}$ but $f_d^+ = D_2 R_d$}
Given an external force $f_L:TQ\to T^*Q$, the associated 1-form $\alpha$ on $TQ$ is given by
% \footnote{The minus sign has been included to compensate the different sign criteria from Refs.~\cite{de_leon_symmetries_2021,de_leon_geometric_2022,lopez-gordon_geometry_2021,leon_methods_1989} and \cite{marsden_discrete_2001,lew_variational_2004}.}
\begin{equation}
  \alpha = -(f_{L})_i (q,\dot{q})\dd q^i.
\end{equation}
\begin{remark}
It is worth noting that we have changed the sign criteria for discrete forces with respect to our previous papers
% \footnote{For the sake of anonymity we have not cited them in this version of the manuscript.},
 \cite{de_leon_symmetries_2021,lopez-gordon_geometry_2021,de_leon_geometric_2022},
 in order to be consistent with Lew, Marsden, Ortiz and West's criteria \cite{marsden_discrete_2001,lew_variational_2004}.
\end{remark}
The \emph{Poincaré-Cartan 1-form} on $TQ$ associated with the Lagrangian function $L:TQ\to \RR$ is
\begin{equation}
  \theta_L=S^*(\dd L),
\end{equation}
where $S^*$ is the adjoint operator of the vertical endomorphism on $TQ$, which is locally
\begin{equation}
  S = \dd q^i \otimes \frac{\partial  } {\partial  \dot{q}^i}.
\end{equation}
The \emph{Poincaré-Cartan 2-form} is $\omega_L=-\dd \theta_L$, so locally
\begin{equation}
  \omega_L = \dd q^i \wedge \dd \left( \frac{\partial L} {\partial  \dot{q}^i}  \right).
\end{equation}
One can easily verify that $\omega_L$ is symplectic if and only if $L$ is regular, that is, if the Hessian matrix
\begin{equation}
  \left( W_{ij}  \right) = \left( \frac{\partial^2 L  } {\partial \dot{q}^i \partial \dot{q}^j }  \right)
\end{equation}
is invertible. 
The \emph{energy} of the system is given by
\begin{equation}
  E_L=\Delta(L)-L,
\end{equation}
where $\Delta$ is the Liouville vector field,
\begin{equation}
  \Delta= \dot{q}^i \frac{\partial  } {\partial  \dot{q}^i}.
\end{equation}
The dynamics of the system $(L, \alpha)$ is given by
\begin{equation}
  \contr{\xi_{L,\alpha}} \omega_L = \dd E_L + \alpha, \label{dynamics_Lagrangian}
\end{equation}
that is, the integral curves of the \emph{forced Euler-Lagrange vector field} $\xi_{L,\alpha}$ satisfy the forced Euler-Lagrange equations \eqref{forced_continuous_EL}. This vector field is a \emph{second order differential equation} (\emph{SODE}), that is, 
\begin{equation}
  S(\xi_{L,\alpha}) = \Delta.
\end{equation}
An external force $\bar{R}$ is \emph{Rayleigh} (see Refs.~\cite{de_leon_symmetries_2021,lopez-gordon_geometry_2021}, see also Ref.~\cite{bloch_euler-poincare_1996}) if there exists a function $\mathcal{R}$ on $TQ$ such that
\begin{equation}
  \bar{R}=S^*(\dd \mathcal{R}),
\end{equation}
which can be locally written as
\begin{equation}
  \bar{R} = \frac{\partial \mathcal{R}} {\partial \dot{q}^i} \dd q^i.
\end{equation}
This function $\mathcal{R}$ is called the \emph{Rayleigh dissipation function} (or the \emph{Rayleigh potential}). A \emph{Rayleigh system} $(L,\mathcal{R})$ is a forced Lagrangian system with Lagrangian function $L$, and with external force generated by the Rayleigh potential $\mathcal{R}$. The associated fibred morphism is given by minus the fibre derivative of $\mathcal R$, namely, $f_L\colon (q^i, \dot q^i)\mapsto (q^i, -\partial \mathcal R/\partial \dot q^i)$.

The Legendre transformation is the fibre derivative of $L$. That is, it is a mapping $\Leg : T Q \to T^* Q$ which is locally given by
\begin{equation}
  \Leg: (q^i,\dot{q}^i) \mapsto (q^i,p_i),
\end{equation}
where $p_i= \partial L/\partial \dot q^i$.
It is the fibered morphism associated to the $1$-form $\theta_L$. Hence, the following diagram commutes.

\begin{center}
\begin{tikzcd}
TQ \arrow[rr, "\Leg"] \arrow[rd, "\tau_q"'] &   & T^*Q \arrow[ld, "\pi_Q"] \\
                                                          & Q &                      
\end{tikzcd}
\end{center}

% The \emph{Legendre transformation} is a mapping $\Leg: TQ\to T^*Q$ such that the diagram

% \begin{center}
% \begin{tikzcd}
% TQ \arrow[rr, "\Leg"] \arrow[rd, "\tau_q"'] &   & T^*Q \arrow[ld, "\pi_Q"] \\
%                                                           & Q &                      
% \end{tikzcd}
% \end{center}
% commutes. Here $\tau_q$ and $\pi_Q$ are the canonical projections on $Q$.
% Locally,
% \begin{equation}
%   \Leg: (q^i,\dot{q}^i) \mapsto (q^i,p_i),
% \end{equation}
% with $p_i=\partial L/\partial \dot{q}^i$. In other words, $\Leg$ is the fibred morphism associated with the semibasic 1-form $\theta_L$. 
In what follows, let us assume that the Lagrangian $L$ is \emph{hyperregular}, i.e., that $\Leg$ is a (global) diffeomorphism.

\section{Discrete mechanics with external forces}\label{section_discrete}
\subsection{Forced discrete Lagrange-D'Alembert equations}

% \todo[inline]{Forced discrete vs forced discrete}
In the discrete framework, the tangent space $TQ$  and the Lagrangian $L$ are substituted by their discrete counterparts: $Q \times Q$ and $L_{d}: Q \times Q \rightarrow \mathbb{R}$, respectively. 
% Given a smooth manifold $M$, we denote by $\Omega^1(M)$ the real vector space (and $C^\infty$-module) of 1-forms on $M$.
% \todo{\tiny Explicar por qué $\Omega^1(Q \times Q)$ }
We can identify $\Omega^1(Q\times Q)$ with $\Omega^1(Q) \oplus \Omega^1(Q)$, and thus the discrete version of the force $f:TQ\to T^*Q$ can be taken to be a one-form $f_d = (f_d^+, f_d^-)  \in \Omega^1(Q \times Q)$. Instead of curves $q:[0,T]\to Q$, we will consider their discrete versions, which are the $n$-tuples $\left\{q_k  \right\}_{k=0}^N \in Q^{N+1}$. In other words, a curve in $Q$ is understood now as a sequence of $N+1$ of its points (see Figure \ref{fig_discrete_continuous}). As we will see, there is a correspondence between these objects and their continuous counterparts. 
% \todo[inline]{Dibujo}

\begin{figure}[h]
\centering
\includegraphics[width=.5\linewidth]{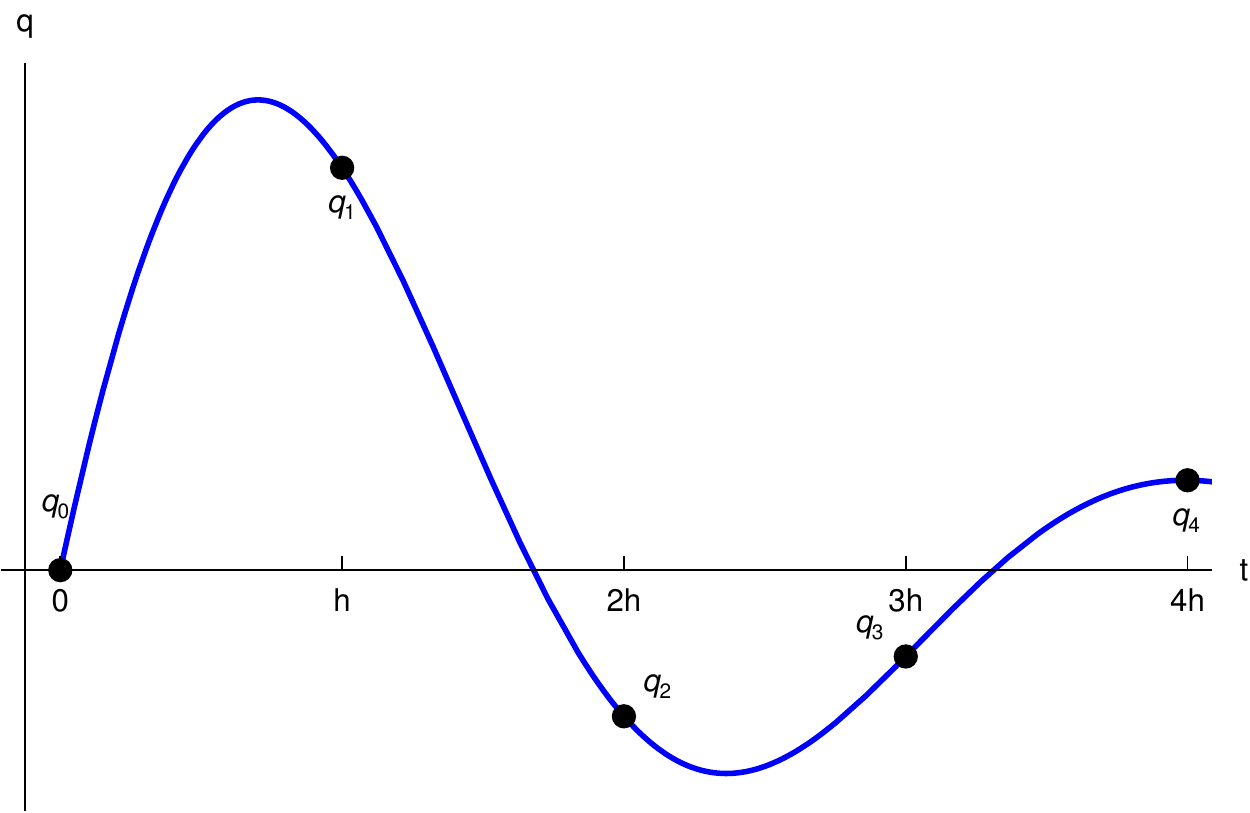}
\caption{Continuous and discrete trajectories of a non-linear simple pendulum with a linear Rayleigh dissipation.}
\label{fig_discrete_continuous}
\end{figure}

The \emph{discrete Lagrange-d'Alembert principle} seeks discrete curves $\left\{q_k  \right\}_{k=0}^N$ that satisfy
\begin{equation}
\delta \sum_{k=0}^{N-1} L_{d}\left(q_{k}, q_{k+1}\right)+\sum_{k=0}^{N-1}\left[f_{d}^{-}\left(q_{k}, q_{k+1}\right) \cdot \delta q_{k}+f_{d}^{+}\left(q_{k}, q_{k+1}\right) \cdot \delta q_{k+1}\right]=0
\end{equation}
for all variations $\left\{\delta q_{k}\right\}_{k=0}^{N} \in \prod_{i=0}^N{T_{q_i} Q}$ vanishing at the endpoints.

This principle provides the \emph{forced discrete Euler-Lagrange equations}
\begin{equation}
D_{2} L_{d}\left(q_{k-1}, q_{k}\right)+D_{1} L_{d}\left(q_{k}, q_{k+1}\right)+f_{d}^{+}\left(q_{k-1}, q_{k}\right)+f_{d}^{-}\left(q_{k}, q_{k+1}\right)=0. \label{discrete_forced_EL}
\end{equation}
The right and left \emph{forced discrete Legendre transforms} are given by 
\begin{equation}
\begin{aligned}
\Legp: Q\times Q &\to T^*Q\\
\left(q_j, q_{j+1}\right) &\mapsto \left(q_{j+1}, D_{2} L_{d}\left(q_j, q_{j+1}\right)+f_{d}^{+}\left(q_j, q_{j+1}\right)\right),
\end{aligned}
\end{equation}
and
\begin{equation}
\begin{aligned}
\Legm: Q\times Q &\to T^*Q\\
\left(q_j, q_{j+1}\right) &\mapsto \left(q_{j}, -D_{1} L_{d}\left(q_j, q_{j+1}\right)-f_{d}^{-}\left(q_j, q_{j+1}\right)\right),
\end{aligned}
\end{equation}
respectively. The corresponding momenta are
\begin{equation}
  p_{j,j+1}^+=D_{2} L_{d}\left(q_j, q_{j+1}\right)+f_{d}^{+}\left(q_j, q_{j+1}\right), \label{momenta_plus}
\end{equation}
and
\begin{equation}
  p_{j,j+1}^- = -D_{1} L_{d}\left(q_j, q_{j+1}\right)-f_{d}^{-}\left(q_j, q_{j+1}\right). \label{momenta_minus}
\end{equation}
By the forced discrete Euler-Lagrange equations \eqref{discrete_forced_EL},
\begin{equation}
  p_{j-1,j}^+ = p_{j,j+1}^- \eqqcolon p_j.
\end{equation}
One can introduce the \emph{discrete Hamiltonian flow}
\begin{equation}
\begin{aligned}
\mathcal{F}_{d}^{H}=\Legp \circ\left(\Legm  \right)^{-1}:T^*Q&\to T^*Q\\
(q_j,p_j) & \mapsto (q_{j+1},p_{j+1}).
\end{aligned}
\end{equation}
This is formally identical to the flow defined in Ref.~\cite{de_leon_geometry_2018}, albeit with the discrete Legendre transforms replaced by the forced discrete Legendre transforms from Ref.~\cite{marsden_discrete_2001}. Similarly, we define the \emph{discrete Lagrangian flow} by
\begin{equation}
\begin{aligned}
  \mathcal{F}_{L_d}^{f_d} =  \left(\Legm  \right)^{-1} \circ \Legp : Q\times Q & \to Q\times Q\\
  (q_ {j-1}, q_j) & \mapsto (q_j, q_{j+1}).  
\end{aligned}
\end{equation}

\subsection{Forced discrete Hamilton equations}
% \todo[inline]{Do they have an intrinsic definition?}

 Let us define the \emph{right discrete Hamiltonian} $H_d^+\colon T^*Q\to \RR$ by
\begin{equation}
H_{d}^{+}\left(q_{j}, p_{j+1}\right)=p_{j+1} q_{j+1}-L_{d}\left(q_{j}, q_{j+1}\right),
\end{equation}
% from where the \emph{forced right discrete Hamilton equations} 
% \begin{equation}
% \begin{aligned}
% &\left[  q_{j+1} -D_{2} H_{d}^{+}\left(q_{j}, p_{j+1}\right)  \right] 
% \frac{\partial p_{j+1}} {\partial q_{j+1} }
%    =       -f_d^+(q_j,q_{j+1}) , \\
% &p_{j} =D_{1} H_{d}^{+}\left(q_{j}, p_{j+1}\right) 
%        - f_d^-\left(q_{j}, p_{j+1}\right), \label{right_discrete_Hamilton}
% \end{aligned}
% \end{equation}
% can be easily derived.
 % Analogously, we have 
 and the \emph{left discrete Hamiltonian} $H_d^-\colon T^*Q\to \RR$ by
\begin{equation}
H_{d}^{-}(q_{j+1},p_{j})=-p_{j} q_{j}-L_{d}\left(q_{j}, q_{j+1}\right).
\end{equation}
% and the \emph{forced left discrete Hamilton equations} 
% \begin{equation}
% \begin{aligned}
% &q_{j} =-D_{2} H_{d}^{-}\left(q_{j+1}, p_{j}\right), \\
% &p_{j+1} =-D_{1} H_{d}^{-}\left(q_{j+1}, p_{j}\right)
%           +f_d^+\left(q_{j+1}, p_{j}\right) .
% \end{aligned}
% \end{equation}
The \emph{discrete action} is 
\begin{equation}
S_{d}^{N}\left(\left\{q_{j}\right\}_{j=0}^{N-1}\right)=\sum_{j=0}^{N-1} L_{d}\left(q_{j}, q_{j+1}\right).
\end{equation}
In terms of the right discrete Hamiltonian, it can be written as
\begin{equation}
  S_{d}^{N}\left(\left\{q_{j}\right\}_{j=0}^N\right)
  = \sum_{j=0}^{N-1} \left[ 
      p_{j+1} q_{j+1}-H_{d}^{+}\left(q_{j}, p_{j+1}\right)
    \right], \label{discrete_action_Hamiltonian}
\end{equation}
so the discrete Lagrange-d'Alembert principle can be rewritten as
\begin{equation}
\delta \sum_{j=0}^{N-1}
\left[ p_{j+1}q_{j+1}-H_d^+(q_j,p_{j+1})  \right]
+\sum_{j=0}^{N-1}\left[f_{d}^{-}\left(q_{j}, q_{j+1}\right) \cdot \delta q_{j}+f_{d}^{+}\left(q_{j}, q_{j+1}\right) \cdot  \delta q_{j+1}\right]=0,
\end{equation}
from where the \emph{forced right discrete Hamilton equations} 
\begin{subequations}
\begin{flalign}
&\left[  q_{j+1} -D_{2} H_{d}^{+}\left(q_{j}, p_{j+1}\right)  \right] 
\frac{\partial p_{j+1}} {\partial q_{j+1} }
   =       -f_d^+(q_j,q_{j+1}), 
   \label{right_discrete_Hamilton_a} \\
&p_{j} =D_{1} H_{d}^{+}\left(q_{j}, p_{j+1}\right) 
       - f_d^-\left(q_{j}, q_{j+1}\right), 
       \label{right_discrete_Hamilton_b}
\end{flalign} \label{right_discrete_Hamilton}
\end{subequations}
can be easily derived. Similarly, we can obtain the \emph{forced left discrete Hamilton equations}
\begin{subequations}
\begin{flalign} 
  &\left[  q_{j} +D_{2} H_{d}^{-}\left(q_{j+1}, p_{j}\right)  \right] 
  \frac{\partial p_{j}} {\partial q_{j} }
   =       f_d^-(q_j,q_{j+1}) , \label{left_discrete_Hamilton_a}\\
  &p_{j+1} =-D_{1} H_{d}^{-}\left(q_{j+1}, p_{j}\right) 
       + f_d^+\left(q_{j}, q_{j+1}\right). 
       \label{left_discrete_Hamilton_b}
\end{flalign} \label{left_discrete_Hamilton}
\end{subequations}

As we shall explain in Section \ref{section_HJ}, the solutions of the right (resp.~left) discrete Hamilton equations are related with the solutions of the so-called right (resp.~left) discrete Hamilton-Jacobi equations via a discrete flow on $Q\times Q$.

\subsection{Discrete Rayleigh forces}
% We say that a discrete force $f_d=(f_d^-,f_d^+)$ is \emph{Rayleigh} if $\dd f_d = 0$, that is, locally $f_d = \dd \mathcal{R}_d$ for some $\mathcal{R}_d : Q \times Q \to \mathbb{R}$, which we call the \emph{discrete Rayleigh potential}. (or is it $f_d = (- D_0 f^- , D_1 f^+))$.
% \todo[inline,color=red!50]{¡Ojo! He cambiado el criterio de signos para ser consistentes con el de Marsden y West (esto se traduce en intercambiar $L_d^+\leftrightarrow L_d^-$).}
\begin{definition}
We say that a discrete force $f_d=(f_d^-,f_d^+)$ is \emph{Rayleigh} if there exists a function $\mathcal{R}_d : Q \times Q \to \mathbb{R}$ such that
\begin{equation}
  f_d^+(q_0,q_1)= -D_2\mathcal{R}_d(q_0,q_1),
\end{equation}
and
\begin{equation}
  f_d^-(q_0,q_1)=D_1\mathcal{R}_d(q_0,q_1)
\end{equation}
for each $(q_0,q_1)\in Q\times Q$.
We shall call this function $\mathcal{R}_d$ the \emph{discrete Rayleigh potential}.
Let us also introduce the \emph{modified discrete Lagrangians} $L_d^+=L_d+\mathcal{R}_d$ and $L_d^-=L_d-\mathcal{R}_d$. 
\end{definition}

The forced discrete Euler-Lagrange equations can be written in terms of the modified discrete Lagrangians as
\begin{equation}
D_{2}  L_{d}^-\left(q_{k-1}, q_{k}\right)+D_{1}  L_{d}^+\left(q_{k}, q_{k+1}\right)=0.
 \label{discrete_forced_EL_Rayleigh}
\end{equation}
The \emph{forced discrete Legendre transforms} can be written as
\begin{equation}
\begin{aligned}
  &\Legp: \left(q_j, q_{j+1}\right) &\mapsto \left(q_{j+1}, D_{2}  L_{d}^-\left(q_j, q_{j+1}\right)\right),\\
  &\Legm: \left(q_j, q_{j+1}\right) &\mapsto \left(q_{j}, -D_{1} L_{d}^+\left(q_j, q_{j+1}\right)\right),
\end{aligned}
\end{equation}
so their corresponding momenta are
\begin{equation}
\begin{aligned}
  &p_{j,j+1}^+=D_{2}  L_{d}^-\left(q_j, q_{j+1}\right),\\
  &p_{j,j+1}^- = -D_{1}  L_{d}^+\left(q_j, q_{j+1}\right).
\end{aligned}
\end{equation}
% \todo[inline, color=blue!40]{Sé que la notación aquí puede resultar confusa pero, tal y como hemos definido ahora $\mathcal R_d$, $p_{j,j+1}^\pm$ va con $L_d^\mp$ (cf. Eqs.~\eqref{momenta_plus} y \eqref{momenta_minus}). Podríamos tomar otro convenio de signos para $f_d$ y $\mathcal{R}_d$, pero creo que es mejor adaptarnos al de Marsden-West.}
The forced right discrete Hamilton equations are now
\begin{equation}
\begin{aligned}
&\left[  q_{j+1} -D_{2} H_{d}^{+}\left(q_{j}, p_{j+1}\right)  \right] 
D_2^2 L_d^-(q_j,q_{j+1})
   =       -D_2\mathcal{R}_d(q_j,q_{j+1}) , \\
&p_{j} =D_{1} H_{d}^{+}\left(q_{j}, p_{j+1}\right) 
       + D_1\mathcal{R}_d(q_j,q_{j+1}), 
\end{aligned}
\end{equation}
and the forced left discrete Hamilton equations are
\begin{equation}
\begin{aligned} 
  &\left[  q_{j} +D_{2} H_{d}^{-}\left(q_{j+1}, p_{j}\right)  \right] 
  D_1^2 L_d^+ (q_j, q_{j+1})
   =  D_1 {R}_d (q_j, q_{j+1}) , \\
  &p_{j+1} =-D_{1} H_{d}^{-}\left(q_{j+1}, p_{j}\right) 
       + D_2 {R}_d \left(q_{j}, q_{j+1}\right). \label{left_discrete_Hamilton_Rayleigh}
\end{aligned}
\end{equation}
Here $D_i^2=D_i\circ D_i$, for $i=1,2$.

\begin{proposition}[Equivalent discrete Rayleigh systems]
Consider two discrete Rayleigh systems $(L_d, \mathcal{R_d})$ and $(\tilde{L}_d, \tilde{\mathcal{R}}_d)$, with $\tilde{L}_d= L_d + \phi$ and $\tilde{\mathcal{R}}_d= \mathcal{R}_d + \chi$ for some functions $\phi$ and $\chi$ on $Q\times Q$. Then, $(L_d, \mathcal{R_d})$ and $(\tilde{L}_d, \tilde{\mathcal{R}}_d)$ are \emph{equivalent} (i.e., they lead to the same forced discrete Euler-Lagrange equations \eqref{discrete_forced_EL_Rayleigh}) if and only if 
\begin{subequations}
\begin{equation}
  \tilde{L}_d(q_0, q_1) 
  = L_d (q_0, q_1) + \psi(q_0) + \varphi_1(q_1) 
  + \varphi_0(q_0) - \psi(q_1)  ,
\end{equation}
and
\begin{equation}
  \tilde{\mathcal{R}}_d(q_0, q_1) 
  = \mathcal{R}_d (q_0, q_1)  + \psi(q_0) - \varphi_1(q_1) 
  - \varphi_0(q_0) + \psi(q_1)  ,
\end{equation}
\end{subequations}
for some functions $\psi, \varphi_0, \varphi_1$ on $Q$.
% and some constants $c, d$. 

In other words, $(L_d^+, L_d^-)$ and $(\tilde L_d^+, \tilde L_d^-)$ are equivalent if and only if
\begin{subequations}
\begin{equation}
  \tilde L_d^+(q_0, q_1) = L_d^+(q_0, q_1) + 2\psi (q_0) +2\varphi_1 (q_1),
\end{equation}
and
\begin{equation}
  \tilde L_d ^-(q_0, q_1) = L_d^-(q_0, q_1) + 2\varphi_0(q_0) - 2\psi(q_1) .
\end{equation}
\end{subequations}

\end{proposition}

\begin{proof}
Let $L_d^\pm = L_d \pm \mathcal R_d$.
Let $\tilde{L}_d^\pm = L_d^\pm + f^\pm$ for some functions $f^+, f^-$ on $Q\times Q$. Clearly, $(L_d^-, L_d^+)$ and $(\tilde{L}_d^-, \tilde{L}_d^+)$ lead to the same forced discrete Euler-Lagrange equations \eqref{discrete_forced_EL_Rayleigh} if and only if
\begin{equation}
  D_2 f^- (q_{k-1}, q_k) + D_1 f^+(q_k, q_{k+1}) = 0,  
  \label{condition_equivalent_Rayleigh_discrete}
\end{equation}
which implies that
\begin{equation}
  D_1 D_2 f^-(q_{k-1}, q_{k}) = D_1 D_2 f^+ (q_k, q_{k+1}) = 0,
\end{equation}
and thus
\begin{equation}
  f^\pm (q_0, q_{1}) = \varphi_0^\pm (q_0) + \varphi^\pm_1 (q_1) 
  % + c^\pm
\end{equation}
for some functions $\varphi_0^\pm, \varphi_1^\pm$ on $Q$.
% and some constants $c^\pm$.
 Now, by Eq.~\eqref{condition_equivalent_Rayleigh_discrete}, we have that
\begin{equation}
 \left(\varphi^-_1 \right)'(q_k) + \left(\varphi^+_0  \right)'(q_k) = 0,
\end{equation}
so
\begin{equation}
  \varphi^+_0 (q_k) = -\varphi^-_1(q_k) + b,
  % \eqqcolon \psi(q_k),
\end{equation}
for some constant $b$. Let us denote $\psi=1/2\ \varphi^+_0$, $\varphi_0= 1/2\ \varphi_0^-$ and $\varphi_1=1/2\ \varphi_1^+$.
% let $c=1/2\ c^+$ and $d=1/2\ (c^-+b)$. 
Then, we have that
\begin{equation}
  f^+(q_0, q_1) = 2\psi(q_0) + 2\varphi_1(q_1),
   % +  2c,
\end{equation}
and
\begin{equation}
  f^-(q_0, q_1) =  2\varphi_0 (q_0) - 2\psi(q_1) 
  +b.
  % + 2d.
\end{equation}
The constant $b$ can be absorbed in $\varphi_0$.
Therefore $(L_d^-, L_d^+)$ and $(\tilde{L}_d^-, \tilde{L}_d^+)$ lead to the same forced discrete Euler-Lagrange equations if and only if
\begin{subequations}
\begin{equation}
  \tilde L_d^+(q_0, q_1) = L_d^+(q_0, q_1) + 2\psi (q_0) +2\varphi_1 (q_1) ,
  % + 2c,
\end{equation}
and
\begin{equation}
  \tilde L_d ^-(q_0, q_1) = L_d^-(q_0, q_1) + 2\varphi_0(q_0) - 2\psi(q_1) ,
  % + 2d,
  % +b,
\end{equation}
\end{subequations}
for some functions $\psi, \varphi_0, \varphi_1$ on $Q$ 
% and some constants $c, d$.
and some constant $b$.
 Obviously, $\tilde{L}_d^\pm$ are the modified Lagrangians associated with $(\tilde L_d, \mathcal{R}_d)$ if and only if
\begin{equation}
  \tilde L_d = \frac{1}{2} \left(\tilde L_d^+ + \tilde L_d^-  \right),
\end{equation}
and
\begin{equation}
  \tilde{\mathcal{R}}_d = \frac{1}{2} \left(\tilde L_d^+ - \tilde L_d^-  \right),
\end{equation}
from where the result follows.
% In terms of the discrete Lagrangian and the discrete Rayleigh potential this equivalence can be expressed as follows.
\end{proof}

\subsection{Exact Lagrangians}
The discrete Lagrangian is an approximation to the integral
\begin{equation}
L_{d}^{\mathrm{ex}}\left(q_{0}, q_{1}, h\right)=\int_{0}^{h} L(q(t), \dot{q}(t))\  \dd t,
\end{equation}
called the \emph{exact discrete Lagrangian}. Here $q:\left[t_{0}, t_{1}\right] \rightarrow Q$ is the solution of the forced Euler-Lagrange equations \eqref{forced_continuous_EL} with boundary conditions $q\left(0\right)=q_{0}, q\left(h\right)=q_{1}$, and $h\in \RR$ is a fixed time step. Similarly, the \emph{discrete forces} \cite{marsden_discrete_2001} are fibre preserving maps $f_d^+, f_d^-:Q\times Q\to T^*Q$, in the sense that $\pi_Q\circ f_d^-=\pi_1$ and $\pi_Q\circ f_d^+=\pi_2$, where $\pi_1,\pi_2:Q\times Q\to Q$ are the natural projections given by $\pi_1(q_i,q_j)=q_i$ and $\pi_2(q_i,q_j)=q_j$. 
The discrete forces are approximations to the \emph{exact discrete forces}: 
\begin{equation}
\begin{aligned}
&f_{d}^{\mathrm{ex}+}\left(q_{0}, q_{1}, h\right)=\int_{0}^{h} f_{L}(q(t), \dot{q}(t)) \cdot \frac{\partial q(t)}{\partial q_{1}} \mathrm{~d} t, \\
&f_{d}^{\mathrm{ex}-}\left(q_{0}, q_{1}, h\right)=\int_{0}^{h} f_{L}(q(t), \dot{q}(t)) \cdot \frac{\partial q(t)}{\partial q_{0}} \mathrm{~d} t.
\end{aligned}
\end{equation}
It can be shown \cite{marsden_discrete_2001} that the exact discrete system is equivalent to the corresponding continuous system. More specifically, the solutions $q:\left[0, t_{N}\right] \rightarrow Q$ of the forced Euler-Lagrange equations for $L$ and solutions $\left\{q_{k}\right\}_{k=0}^{N}$ of the forced discrete Euler-Lagrange equations for $L_{d}^{\mathrm{ex}}$ are related by
\begin{equation}
\begin{aligned}
q_{k} &=q\left(t_{k}\right) \text { for } k=0, \ldots, N, \\
q(t) &=q_{k, k+1}(t) \text { for } t \in\left[t_{k}, t_{k+1}\right] .
\end{aligned}
\end{equation}
Here the curves $q_{k, k+1}:\left[t_{k}, t_{k+1}\right] \rightarrow Q$ are the unique solutions of the Euler-Lagrange equations for $L$ satisfying $q_{k, k+1}(k h)=q_{k}$ and $q_{k, k+1}((k+1) h)=$ $q_{k+1}$.

\begin{remark}[Existence of a discrete Rayleigh potential]
Clearly, a discrete force $f_d$ on $Q\times Q$ is Rayleigh if and only if
\begin{equation}
  D_1 f_d^+ = -D_2 f_d^-.
\end{equation}
For an exact discrete force $f_d^{\mathrm{ex}}=(f_d^{\mathrm{ex}+}, f_d^{\mathrm{ex}-})$ associated with a forced continuous Lagrangian system $(L, f_L)$, this condition can be written as
\begin{equation}
\begin{aligned}
  \int_0^h &\left[2\frac{\partial f} {\partial q}\left(q(t), \dot q(t)  \right) \frac{\partial q(t)} {\partial q_0} \frac{\partial q(t)} {\partial q_1} 
  + \frac{\partial f} {\partial \dot q}\left(q(t), \dot q(t)  \right) \frac{\partial \dot q(t)} {\partial q_0} \frac{\partial q(t)} {\partial q_1} 
  \right. \\ &\left.
  + \frac{\partial f} {\partial \dot q}\left(q(t), \dot q(t)  \right) \frac{\partial q(t)} {\partial q_0} \frac{\partial \dot q(t)} {\partial q_1} 
  + 2 f_L \left(q(t), \dot q(t)  \right) \frac{\partial^2 q(t) } {\partial q_0 q_1}
   \right]\ \dd t = 0,
\end{aligned} \label{condition_discrete_Rayleigh_integral}
\end{equation}
where $q(t)$ is the solution of the forced Euler-Lagrange equations for $(L, f_L)$ with boundary conditions $q(0)=q_0$ and $q(h)=q_1$.
\end{remark}

Motivated by the computation of several examples with Mathematica but unable to find a proof, we claim the following statement.
\begin{conjecture}[relation between continuous and discrete Rayleigh systems] \label{conjecture_Rayleigh}
Consider a continuous Rayleigh system $(L, \mathcal{R})$ on $TQ$, and assume that $L$ is natural, i.e., 
\begin{equation}
  L(q, \dot{q}) = \frac{1}{2} g(\dot{q}, \dot{q}) - V(q)
\end{equation}
where $g$ and $V$ are a Riemannian metric and a function on $Q$, respectively. Let $f_L = - S^* (\dd \mathcal{R})$ be the continuous Rayleigh force. Then, the exact discrete force $f_d^E$ on $Q\times Q$ associated with $f_L$ is Rayleigh.
\end{conjecture}

\begin{example}[Harmonic oscillator with Rayleigh dissipation]\label{example_harmonic_oscillator}
Consider a 1-dimensional harmonic oscillator,
\begin{equation}
  L= \frac{1}{2}m \dot{q}^2 - \frac{1}{2}kq^2,
\end{equation}
with Rayleigh potential
\begin{equation}
  \mathcal{R}=\frac{r}{2} \dot{q}^2.
\end{equation}
Then the forced Euler-Lagrange equations yield
\begin{equation}
  m\ddot{q}+r\dot{q}+kq=0. \label{harmonic_oscillator_Rayleigh_ODE}
\end{equation}
Suppose that $4km>r^2$, and let
\begin{equation}
  a \coloneqq \frac{r}{2m},\quad 
  b \coloneqq \frac{\sqrt{4km-r^2}}{2m}.
  % b\coloneqq \frac{\sqrt{\abs{r^2-4km}}}{2m}.
\end{equation}
The solution of the ODE \eqref{harmonic_oscillator_Rayleigh_ODE} with boundary values $q(0)=q_0$ and $q(h)=q_1$ is then
\begin{equation}
  q(t) = e^{-a t} \left[ 
      q_0 \cos(b t)  + c \sin(b t) 
    \right],
\end{equation}
where
\begin{equation}
  c = \frac{e^{a h}q_1-\cos(bh) q_0}{\sin(bh)}.
\end{equation}
The exact discrete Lagrangian is then
% \begin{equation}
% \begin{aligned}
%   L_d(q_0,q_1) &= 
%   \frac{e^{-2 a h} }{8 a \left(a^2+b^2\right)}
%   \left\{2 a b e^{2 a h} \left(q_0^2+q_1^2\right) \cot (b h) 
%   \left(m \left(a^2+b^2\right)+k\right)
%    \right.\\&\left.
%   +\left(m \left(a^2+b^2\right)-k\right) 
%   \left[2 a^2 e^{2 a h} (q_0-q_1) (q_0+q_1)+b^2 \left(e^{2 a h}-1\right) \csc ^2(b h) 
%    \right.\right. \\&\left.\left.
%   \left(-2 q_0 q_1 e^{a h} \cos (b h)+q_1^2 e^{2 a h}+q_0^2\right)\right]
%   -2 a b q_0 q_1 e^{a h} \left(e^{2 a h}+1\right) \csc (b h) \left(m \left(a^2+b^2\right)+k\right)\right\},
% \end{aligned}
% \end{equation}
\begin{equation}
\begin{aligned} 
  L_d
  =&
  \frac{1}{16} (\coth (b h)-1)^2 \left[\frac{2 \left(m \left(b^2-a^2\right)+k\right) \left(q_0  e^{2 b h}-q_1  e^{h (a+b)}\right) \left(q_0 -q_1  e^{h (a+b)}\right)}{a}
  %%%%%%%%%%%%%%%
  \right.\\&\left.
  %%%%%%%%%%%%%%%
  +e^{-2 h (a+b)} \left(-\frac{2 e^{3 b h} \left(m \left(b^2-a^2\right)+k\right) \left(q_0  e^{b h}-q_1  e^{a h}\right) \left(q_0 -q_1  e^{h (a+b)}\right)}{a}
  %%%%%%%%%%%%%%%
  \right.\right.\\&\left.\left.
  %%%%%%%%%%%%%%%
  +\frac{e^{4 b h} \left(k-m (a-b)^2\right) \left(q_0 -q_1  e^{h (a+b)}\right)^2}{a-b}+\frac{\left(k-m (a+b)^2\right) \left(q_0  e^{2 b h}-q_1  e^{h (a+b)}\right)^2}{a+b}\right)
  %%%%%%%%%%%%%%%
  \right.\\&\left.
  %%%%%%%%%%%%%%%
  +\frac{\left(m (a-b)^2-k\right) \left(q_0 -q_1  e^{h (a+b)}\right)^2}{a-b}+\frac{\left(m (a+b)^2-k\right) \left(q_0  e^{2 b h}-q_1  e^{h (a+b)}\right)^2}{a+b}\right]
\end{aligned}
\end{equation}
and the exact discrete forces are
\begin{equation}
  f_d^+ (q_0,q_1) = \frac{1}{2} r \left(\frac{b q_0 \sinh (a h) \csc (b h)}{a}-q_1\right),
\end{equation}
and
\begin{equation}
  f_d^-(q_0,q_1) = \frac{r (a q_0-b q_1 \sinh (a h) \csc (b h))}{2 a}.
\end{equation}
One can check that condition \eqref{condition_discrete_Rayleigh_integral} holds. As a matter of fact, 
\begin{equation}
  \mathcal{R}_d (q_0,q_1) = \frac{1}{4} r \left(q_0^2+q_1^2\right) -\frac{b q_0 q_1 r \sinh (a h) \text{csch}(b h)}{2 a}
\end{equation}
is a discrete Rayleigh potential from which $f_d$ can be derived.
\end{example}

\subsection{Midpoint rule}
Consider a forced Lagrangian system $(L, f_L)$ on $TQ$ and assume that is regular. The associated \emph{midpoint rule discrete Lagrangian} $L_d$ on $Q\times Q$ is then given by \cite{marsden_discrete_2001}
\begin{equation}
  L_d^{\frac{1}{2}}\left( q_0, q_1, h  \right) = h L \left( \frac{q_0 + q_1}{2}, \frac{q_1-q_0}{h}  \right).
\end{equation}
Similarly, the \emph{midpoint rule discrete forces} on $Q\times Q$ are given by
\begin{equation}
  f_d^{\frac{1}{2}+}\left( q_0, q_1, h  \right) 
  = f_d^{\frac{1}{2}-}\left( q_0, q_1, h  \right) 
  = \frac{h}{2} f_L \left( \frac{q_0 + q_1}{2}, \frac{q_1-q_0}{h}  \right).
\end{equation}

The midpoint rule discrete force is Rayleigh if and only if
\begin{equation}
  D_1 f_d^{\frac{1}{2}\pm }(q_0, q_1) = -D_2 f_d^{\frac{1}{2}\pm }(q_0, q_1),
\end{equation}
which holds if and only if
\begin{equation}
  D_1 f_L (q, \dot q) = 0,
\end{equation}
or, in other words, if $f_L$ is homogeneous.

\begin{remark}
Let $(L, \mathcal{R})$ be a discrete Rayleigh system on $TQ$ and assume that $\mathcal{R}$ is homogeneous (i.e., $\partial \mathcal R/\partial q$=0). Then the associated midpoint rule discrete force is Rayleigh. As a matter of fact, we can define the \emph{midpoint rule discrete Rayleigh potential} $\mathcal R_d^{1/2}$, given by
\begin{equation}
  \mathcal R_d^{\frac{1}{2}}(q_0, q_1) = \frac{h}{2} \mathcal{R}\left(\dot{q}= \frac{q_1-q_0}{h} \right).
\end{equation}

\end{remark}

\begin{example} \label{example_Marsden_West}
Consider a Rayleigh system $(L, \mathcal{R})$ on $T\RR^2$, with
\begin{equation}
  L = \frac{1}{2} \norm{\dot q}^2 - \norm{q}^2 \left( \norm{q}^2 -1  \right)^2, 
\end{equation} 
and
\begin{equation}
  \mathcal{R} = \frac{1}{2} k \norm{\dot q}^2,
\end{equation}
for some constant $k$. Here $q=(q^1, q^2)$ are the Cartesian coordinates in $\RR^2$, $(q, \dot q )= (q^1, q^2 , \dot q^1, \dot q^2)$ are the induced fibred coordinates in $T\RR^2$, and $\norm{\cdot}$ denotes the Euclidean norm in $\RR^2$.

% Clearly, this system is invariant under $\mathrm{SO}(2)\cong \mathbb{S}^1$ acting by rotations. As a matter of fact, $\xi_{\RR^2}=q^2 \partial/\partial q^1 - q^1 \partial/\partial q^2$ is a symmetry of the forced Lagrangian (see Refs.~\cite{de_leon_symmetries_2021,lopez-gordon_geometry_2021}) and $\ell = \xi_{\RR^2}^v(L) = q^2 \dot q^1 - q^1 \dot q^2$ is the associated constant of the motion.

Let us take $k=10^{-3}$, so that it corresponds to Example 3.2.3 from Ref.~\cite{marsden_discrete_2001}. In Figure \ref{fig_energy_Marsden_West} we plot the evolution of the energy of the system for the initial conditions $q_0^1=0$, $\dot{q}_0=\left(1/2, 0 \right)$ and $E_L(q_0, \dot{q}_0)=11/40$. We compare the variational midpoint rule and the standard fourth-order Runge-Kutta method (see for instance \cite{newman_computational_2013,butcher_john_c_rungekutta_2016}) with a benchmark numerical integration of high precision. Observe that the variational midpoint rule reproduces the energy dissipation correctly, whereas the Runge-Kutta or other standard integrators do not. These effect is specially relevant when the external force is small compared to the magnitude of the conservative dynamics and the time period of integration (see Ref.~\cite{hairer_invariant_1999}).

\begin{figure}[h]
\centering
\includegraphics[width=.5\linewidth]{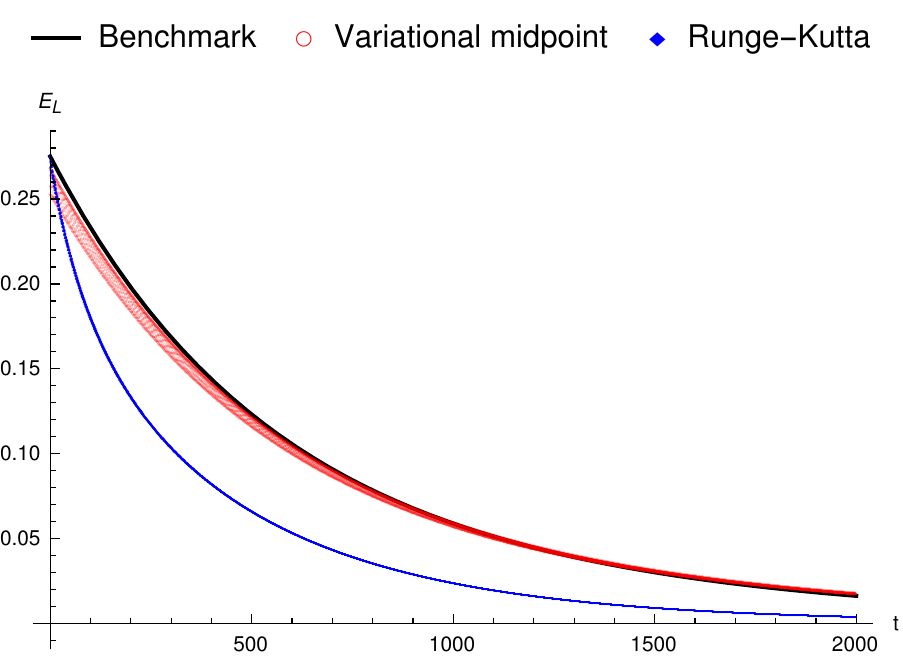}
\caption{Energy of a Rayleigh system computed with the variational midpoint and fourth-order Runge-Kutta methods. Observe the remarkable supremacy of the former, despite being a lower order method.}
\label{fig_energy_Marsden_West}
\end{figure}
\end{example}

\subsection{Discrete Noether's theorem}
% \todo[inline, color=green!40]{Ejemplo(s)}

Let $X$ be a vector field on $Q$. The \emph{complete lift} of $X$ to $Q\times Q$ is the vector field $X^c$ on $Q\times Q$ given by
\begin{equation}
  X^c (q_0, q_1) = \left(X(q_0), X(q_1)  \right)
\end{equation}
for each $(q_0,q_1)\in Q\times Q$. 

Consider the left action $\Phi:G\times Q \to Q$ of a Lie group $G$ on $Q$, and let the vector field $\xi_Q$ on $Q$ be the infinitesimal generator of the action on $Q$. This group action can be lifted to $Q\times Q$ by the product
\begin{equation}
  \Phi_g^{Q\times Q} (q_0, q_1) = \left(\Phi_g (q_0), \Phi_g (q_1)  \right).
\end{equation}
The infinitesimal generator of this action is the vector field $\xi_{Q\times Q} = \xi_Q^c$ on $Q\times Q$.
%  given by
% \begin{equation}
%   \xi_{Q\times Q} (q_0, q_1) = \left(\xi_Q (q_0), \xi_Q (q_1)  \right).
% \end{equation}

Let $\mathfrak{g}$ denote the Lie algebra of $G$, and $\mathfrak{g}^*$ its dual.
Let us introduce the \emph{discrete momentum maps} $J_{L_d}^{f+},J_{L_d}^{f-}:Q\times Q\to \mathfrak{g}^*$ given by
\begin{equation}
\begin{aligned}
  & \left\langle J_{L_d}^{f+}(q_0, q_1), \xi   \right\rangle 
  = \left\langle \Legp (q_0,q_1), \xi_Q (q_1) \right\rangle,   \\
  & \left\langle J_{L_d}^{f-}(q_0, q_1), \xi   \right\rangle 
  = \left\langle \Legm (q_0,q_1), \xi_Q (q_0) \right\rangle . 
\end{aligned}
\end{equation}
When $\langle J_{L_d}^{f+}, \xi  \rangle = \langle J_{L_d}^{f-}, \xi  \rangle $, for some $\xi \in \mathfrak g$, we can define the function
\begin{equation}
    \begin{aligned}
      J_d^\xi:Q\times Q &\to \RR \\
      (q_0, q_1) & \mapsto \left\langle J_{L_d}^{f\pm}, \xi  \right\rangle  (q_0, q_1).
    \end{aligned}
\end{equation}  

% Let $(L_d, f_d)$ be a forced discrete Lagrangian system on $Q\times Q$.

% \begin{definition}
% A vector field $X$ on $Q$ is called a \emph{symmetry of the forced discrete Lagrangian} for $(L_d, f_d)$ if
% \begin{equation}
%   X^d(L_d) + f_d(X^d) = 0.
% \end{equation}
% \end{definition}

\begin{theorem}[Discrete forced Noether's theorem]\label{theorem_Noether}
Let $(L_d, f_d)$ be a forced discrete Lagrangian system on $Q\times Q$. Let $G$ be a Lie group acting on $Q$ and let $\mathfrak{g}$ be the Lie algebra of $G$. Then $J_d^\xi$ is a constant of the motion 
% (i.e., it is preserved by the discrete Lagrangian flow, $J_d^\xi = \langle J_d^{f\pm} \circ \mathcal{F}_{L_d}^{f_d}, \xi\rangle$) 
if and only if
\begin{equation}
  \xi_{Q\times Q}(L_d) + f_d(\xi_{Q\times Q}) = 0,
\end{equation}
for some $\xi \in \mathfrak g$.
% Then, the following statements are equivalent:
% \begin{enumerate}
% \item For every $\xi\in \mathfrak g$, 
% % $\xi_{Q}$ is a symmetry of the forced discrete Lagrangian. 
% \begin{equation}
%   \xi_{Q\times Q}(L_d) + f_d(\xi_{Q\times Q}) = 0.
% \end{equation}
% \item 
% The two discrete momentum maps coincide, namely $J_{L_d}^{f+} = J_{L_d}^{f-} \eqqcolon J_{L_d}^{f}$, and the discrete momentum map $J_{L_d}^{f}$ is preserved by the discrete Lagrangian flow $\mathcal{F}_{L_d}^{f_d}$.
% \end{enumerate}
\end{theorem}

\begin{proof}
We have that
\begin{equation}
\begin{aligned}
  \left\langle \dd L_d + f_d, \xi_{Q\times Q}  \right\rangle (q_0, q_1)
  % = D_1 L_d (q_0, q_1) \cdot \xi_Q(q_0) + D_2 L_d (q_0, q_1) \cdot \xi_Q(q_1)  
  % + f_d^- (q_0, q_1) \cdot \xi_Q(q_0) + f_d^+ (q_0, q_1) \cdot \xi_Q(q_1)
 & = \left(D_1 L_d+ f_d^-  \right)  (q_0, q_1) \cdot \xi_Q(q_0) 
    + \left(D_2 L_d+ f_d^+  \right)  (q_0, q_1) \cdot \xi_Q(q_1) \\
 & = \Legp(q_0, q_1) \cdot \xi_q(q_1) - \Legm (q_0, q_1) \cdot \xi_Q(q_0) \\
 & = \left(J_{L_d}^{f+} - J_{L_d}^{f-}  \right)(q_0,q_1) \cdot \xi,
\end{aligned}
\end{equation}
so  $\langle J_{L_d}^{f+}, \xi \rangle = \langle J_{L_d}^{f-}, \xi \rangle = J^\xi$ if and only if $\xi_{Q\times Q}(L_d) + f_d(\xi_{Q\times Q})$ vanishes.
% so $J_{L_d}^{f+} = J_{L_d}^{f-}$ if and only if $\xi_{Q\times Q}(L_d) + f_d(\xi_{Q\times Q})$ vanishes for every $\xi \in \mathfrak{g}$. 
%%%%%%%%%%%
% Taking variations vanishing at the endpoints of the form $\delta q_k = \xi_Q(q_k)$, 
% for $\xi \in \mathfrak g$, 
Moreover, we have
\begin{equation}
\begin{aligned}
  % \delta & \sum_{k=0}^{N-1} 
  % L_d (q_k, q_{k+1})
  % +\sum_{k=0}^{N-1}\left[f_{d}^{-}\left(q_{k}, q_{k+1}\right) \cdot \delta q_{k}+f_{d}^{+}\left(q_{k}, q_{k+1}\right) \cdot  \delta q_{k+1}\right]\\
  &  \sum_{k=0}^{N-1} \left\langle \dd L_d + f_d, \xi_{Q\times Q}  \right\rangle (q_k, q_{k+1})\\
  & = \sum_{k=1}^{N-1}\left[D_{2} L_{d}\left(q_{k-1}, q_{k}\right)+D_{1} L_{d}\left(q_{k}, q_{k+1}\right)+f_{d}^{+}\left(q_{k-1}, q_{k}\right)+f_{d}^{-}\left(q_{k}, q_{k+1}\right)\right] \cdot \xi_{Q}\left(q_{k}\right)\\
  &\quad +\left[D_{2} L_{d}\left(q_{N-1}, q_{N}\right)+f_{d}^{+}\left(q_{N-1}, q_{N}\right)\right] \cdot \xi_{Q}\left(q_{N}\right)
   +\left[D_{1} L_{d}\left(q_{0}, q_{1}\right)+f_{d}^{-}\left(q_{0}, q_{1}\right)\right] \cdot \xi_{Q}\left(q_{0}\right)\\
  &=\left[D_{2} L_{d}\left(q_{N-1}, q_{N}\right)+f_{d}^{+}\left(q_{N-1}, q_{N}\right)\right] \cdot \xi_{Q}\left(q_{N}\right)
   +\left[D_{1} L_{d}\left(q_{0}, q_{1}\right)+f_{d}^{-}\left(q_{0}, q_{1}\right)\right] \cdot \xi_{Q}\left(q_{0}\right)\\
 & =\Legp \left(q_{N-1}, q_{N}\right) \cdot \xi_{Q}\left(q_{N}\right)-\Legm \left(q_{0}, q_{1}\right) \cdot \xi_{Q}\left(q_{0}\right)\\
 & =\left\langle J_{L_d}^{f+} (q_{N-1}, q_N) - J_{L_d}^{f-} (q_0, q_1), \xi  \right\rangle
   % = \left\langle \left(J_{L_d}^f \circ \mathcal{F}_{L_d}^{f_d} - J_{L_d}^f   \right)(q_0, q_1), \xi \right\rangle.
  = J_d^\xi  (q_{N-1}, q_N) - J_d^\xi (q_0, q_1),
\end{aligned}
\end{equation}
where we have used the forced discrete Euler-Lagrange equations \eqref{discrete_forced_EL}. Therefore $J_d^\xi$ is a constant of the motion (i.e., $J_d^\xi  (q_{N-1}, q_N)= J_d^\xi  (q_{0}, q_1)$) if and only if $\xi_{Q\times Q}(L_d) + f_d(\xi_{Q\times Q})$ vanishes. 
% If his holds for every $\xi \in \mathfrak g$, the Lagrange-d'Alembert principle implies that
% \begin{equation}
%   J_{L_d}^f \circ \mathcal{F}_{L_d}^{f_d} - J_{L_d}^f = 0.
% \end{equation}
% The converse is straightforward.
\end{proof}

% \begin{remark}
% This result was previously found by Marsden and West \cite[Theorem 3.2.1]{marsden_discrete_2001}, however they assumed a stronger hypothesis, namely $\xi_{Q\times Q} (L_d)$ and $f_d(\xi_{Q\times Q})$ vanishing independently 
% % \begin{equation}
% % \begin{aligned}
% %   & \xi_{Q\times Q} (L_d) = 0,  \\
% %   & f_d(\xi_{Q\times Q}) = 0,
% % \end{aligned}
% % \end{equation}
% for all $\xi \in \mathfrak g$. As a matter of fact, this is the discrete analogue of the first statement from our reduction lemma (see Ref.~\cite{de_leon_symmetries_2021}, see also Refs.~\cite{lopez-gordon_geometry_2021,de_leon_geometric_2022}). The discrete analogue of this lemma is as follows.
% \end{remark}

% \todo[inline, color=blue!40]{En el primer paper nos interesaba considerar un álgebra que dejara invariante el lagrangiano, ya que a partir del lagrangiano definíamos la aplicación momento, y luego la subálgebra que dejaba además invariante la fuerza externa. Aquí, como las aplicaciones momento van a depender de $f_d$  además de $L_d$, no veo qué nos aporta pedir que el lagrangiano y la fuerza discretos sean invariantes de forma independiente, pero lo miraré con más detalle.}

% \todo[inline, color=blue!40]{Las simetrías discretas de Noether y demás son bastante más complicadas, ya que no tenemos una forma simpléctica única.}

\begin{theorem}
Let $(L_d, f_d)$ be a forced discrete Lagrangian system on $Q\times Q$. Let $G$ be a Lie group acting on $Q$ and let $\mathfrak{g}$ be the Lie algebra of $G$.
Suppose that $L_d$ is $G$-invariant. Then, for each $\xi\in \mathfrak{g}$,
\begin{enumerate}
\item $J_d^\xi$ is a constant of the motion if and only if 
\begin{equation}
  f_d(\xi_{Q\times Q}) = 0. \label{eq_constant_motion_subalgebra}
\end{equation}
\item If the equation above holds, then $\xi$ leaves $f_d$ invariant if and only if
\begin{equation}
  \contr{\xi_{Q\times Q}} \dd f_d = 0.
\end{equation}
\end{enumerate}
Moreover, the vector subspace $\mathfrak{g}_{f_d}$ of $\mathfrak{g}$ given by
\begin{equation}
  \mathfrak{g}_{f_d} = \left\{\xi \in \mathfrak{g}\mid  f_d(\xi_{Q\times Q}) = 0,\ \contr{\xi_{Q\times Q}} \dd f_d=0   \right\}
\end{equation}
is a Lie subalgebra of $\mathfrak g$.
\end{theorem}

\begin{proof}
Clearly, $L_d$ is $G$-invariant if and only if
\begin{equation}
  \xi_{Q\times Q} (L_d) = 0
\end{equation}
for each $\xi \in \mathfrak{g}$. 
% \todo{Complete this part}
Combining this with Theorem \ref{theorem_Noether}, we have that $J_d^\xi$ is a constant of the motion if and only if 
\begin{equation}
  f_d(\xi_{Q\times Q}) = 0.
\end{equation}
If this equation holds, $f_d$ is $\xi$-invariant if and only if
\begin{equation}
  \liedv{\xi_{Q\times Q}} f_d = \contr{\xi_{Q\times Q}} \dd f_d = 0.
\end{equation}
Let $\xi, \eta\in \mathfrak{g}_{f_d} = \left\{\xi \in \mathfrak{g}\mid  f_d(\xi_{Q\times Q}) = 0,\ \contr{\xi_{Q\times Q}} \dd f_d=0   \right\}$. Then,
\begin{equation}
\begin{aligned}
  f_d \left([\xi_Q, \eta_Q]^c  \right)
  &= f_d \left([\xi_{Q\times Q}, \eta_{Q\times Q}]  \right)
  = \contr{[\xi_{Q\times Q}, \eta_{Q\times Q}]} f_d
  = \liedv{\xi_{Q\times Q}} \contr{\eta_{Q\times Q}} f_d 
    - \contr{\eta_{Q\times Q}} \liedv{\xi_{Q\times Q}} f_d\\
  &= \xi_{Q\times Q}\left(f_d(\eta_{Q\times Q})  \right)
    -\eta_{Q\times Q}\left(f_d(\xi_{Q\times Q})  \right)
    - \contr{\eta_{Q\times Q}} \contr{\xi_{Q\times Q}} \dd f_d
  = 0,
\end{aligned}
\end{equation}
and
\begin{equation}
\begin{aligned}
  \contr{[\xi_Q, \eta_Q]^c} \dd f_d
  &= \contr{[\xi_{Q\times Q}, \eta_{Q\times Q}] } \dd f_d
  = \liedv{\xi_{Q\times Q}} \contr{\eta_{Q\times Q}} \dd f_d 
    - \contr{\eta_{Q\times Q}} \liedv{\xi_{Q\times Q}} \dd f_d\\
  &= \liedv{\xi_{Q\times Q}} \contr{\eta_{Q\times Q}} \dd f_d 
    - \contr{\eta_{Q\times Q}} \dd \liedv{\xi_{Q\times Q}} f_d
  = 0.
\end{aligned}
\end{equation}
Since $\xi\mapsto \xi_Q$ is a Lie algebra antihomomorphism \cite{ortega_momentum_2004}, this proves that $[\xi, \eta]\in \mathfrak{g}_{f_d}$.
\end{proof}

This is the discrete analogue of the first statement from our reduction lemma. \cite[Lemma 15]{de_leon_symmetries_2021} (see also Refs.~\cite{lopez-gordon_geometry_2021,de_leon_geometric_2022}). 
The first statement was previously found by Marsden and West \cite[Theorem 3.2.1]{marsden_discrete_2001}.

\begin{remark}
Let $\xi_Q\in \mathfrak{X}(Q)$ be the infinitesimal generator of the action of $G$ on $Q$. Let us define the vector field $\widehat \xi_{Q\times Q}\in \mathfrak{X}(Q\times Q)$ given by 
\begin{equation}
  \widehat \xi_{Q\times Q}(q_0, q_1) = \left(\xi_Q(q_0), -\xi_Q(q_1)  \right).
\end{equation}
Let $(L_d, \mathcal{R}_d)$ be a discrete Rayleigh system on $Q\times Q$, with associated discrete external force $f_d=(f_d^-, f_d^+)$. Then,
\begin{equation}
\begin{aligned}
    \widehat{\xi}_{Q\times Q} (\mathcal R_d) (q_0, q_1) 
   & = D_1  \mathcal R_d (q_0, q_1) \cdot \xi_Q(q_0)
    - D_2 \mathcal R_d (q_0, q_1) \cdot \xi_Q(q_1)\\
    &= f_d^-(q_0, q_1) \cdot \xi_Q(q_0) + f_d^+(q_0, q_1) \cdot \xi_Q(q_1)
    = f_d(q_0, q_1) \cdot \xi_{Q\times Q}(q_0, q_1).
\end{aligned}
\end{equation}
Hence, $J_d^\xi$ is a constant of the motion  
if and only if
\begin{equation}
  \xi_{Q\times Q}(L_d) +   \widehat{\xi}_{Q\times Q} (\mathcal R_d)  = 0.
\end{equation}
%  Then, Eq.~\eqref{eq_constant_motion_subalgebra} reduces to
% \begin{equation}
%   \widehat \xi_{Q\times Q}(\mathcal{R}_d) = 0.
% \end{equation}
\end{remark}

\begin{example} Consider a discrete Rayleigh system $(L_d, \mathcal R_d)$ on $\RR^2 \times \RR^2$ of the form
\begin{equation}
\begin{aligned}
  &L_d(r_1, \theta_1, r_2, \theta_2) 
    % = \frac{h}{2} \left(\frac{q_2-q_1}{h}  \right)^2
    = \frac{h}{2} \left(\frac{r_2-r_1}{h}  \right)^2
    +\frac{h}{2} \left(\frac{r_1+r_2}{2}  \right)^2 \left(\frac{\theta_2-\theta_1}{h}  \right)^2
    - V \left(\frac{r_1+r_2}{2}  \right)
    ,\\
    &\mathcal R_d (r_1, \theta_1, r_2, \theta_2)  
    = F \left(r_1, r_2, {\theta_1+\theta_2}  \right),
\end{aligned}
\end{equation}
% where $\norm{\cdot}$ denotes the Euclidean norm in $\RR^2$. ç
where $(r_1, \theta_1, r_2, \theta_2)$ are the coordinates in $\RR^2\times \RR^2$ induced by the polar coordinates $(r, \theta)$ in $\RR^2$. For instance, $(L_d, \mathcal R_d)$ could be the midpoint rule discretization of the Rayleigh system $(L, \mathcal R)$ from Example \ref{example_Marsden_West}.

Consider the action of $\mathbb{S}^1$ on $\RR^2$ by rotations. The infinitesimal generator of this action is $\xi_{\RR^2}= \partial/\partial \theta$, so 
\begin{equation}
  \xi_{\RR^2\times \RR^2}= \frac{\partial  } {\partial \theta_1} + \frac{\partial  } {\partial \theta_2}, \qquad
  \widehat \xi_{\RR^2\times \RR^2}= \frac{\partial  } {\partial \theta_1} - \frac{\partial  } {\partial \theta_2}.
\end{equation}
 Clearly, $\xi_{\RR^2 \times \RR^2}(L_d)=0$ and $\widehat \xi_{\RR^2 \times \RR^2}(\mathcal R_d)=0$, so
 \begin{equation}
   J^\xi (r_1, \theta_1, r_2, \theta_2) = \left(\frac{r_1+r_2}{2}  \right)^2 \left(\frac{\theta_2-\theta_1}{h}  \right)
   - D_3 F(r_1, r_2, \theta_1+ \theta_2 )
 \end{equation}
is a constant of the motion. 
In particular, if $D_3 F=0$ the angular momentum is conserved.
\end{example}

\section{Discrete Hamilton-Jacobi theory for systems with\\ external forces}
\label{section_HJ}

% \todo[inline, color=green!40]{Reordenar en subsecciones}

% \subsection{The discrete flow approach}

In this section, we develop a Hamilton-Jacobi theory for forced discrete Hamiltonian systems. Here, the projected object is a discrete flow. See \cite{de_leon_geometry_2018} for an analogous theory for un-forced discrete Hamiltonian systems, see also \cite{ohsawa_discrete_2011}. The main difference with the un-forced theory is that we cannot consider flows on $Q$ and $T^*Q$, but rather on $Q\times Q$ and $T^*(Q\times Q)$.

From Eq.~\eqref{discrete_action_Hamiltonian}, we have
\begin{equation}
  S_{\mathrm{d}}^{k+1}\left(q_{k+1}\right)-S_{\mathrm{d}}^{k}\left(q_{k}\right)=p_{k+1} \cdot q_{k+1}-H_{\mathrm{d}}^{+}\left(q_{k}, p_{k+1}\right) \label{derive_H-J_1}
\end{equation}
which derived with respect to $q_{k+1}$ yields
\begin{equation}
D S_{\mathrm{d}}^{k+1}\left(q_{k+1}\right)=p_{k+1}+\frac{\partial p_{k+1}}{\partial q_{k+1}} \cdot\left[q_{k+1}-D_{2} H_{\mathrm{d}}^{+}\left(q_{k}, p_{k+1}\right)\right]
=p_{k+1} -f_d^+(q_k,q_{k+1}), \label{derive_H-J_2}
\end{equation}
where on the last step we have used the first of the forced right discrete Hamilton equations \eqref{right_discrete_Hamilton}. Analogously, writing the discrete action in terms of the left discrete Hamiltonian and using the forced left discrete Hamilton equations, we can write
\begin{equation}
  DS_d^k(q_k)=p_k+f_d^-(q_k,q_{k+1}).
\end{equation}
Let us introduce the mappings
% Let $\gamma$ be the mapping
\begin{subequations}
\begin{equation}
\begin{aligned} 
    \gamma^+ = DS_d\circ \pi_2 + f_d^+: Q\times Q &\to T^*Q \\
        (q_j,q_{j+1}) &\mapsto (q_{j+1}, p_{j+1}), \label{definition_gamma_plus}
\end{aligned}
\end{equation}
and
\begin{equation}
\begin{aligned} 
    \gamma^- = DS_d\circ \pi_1 - f_d^-: Q\times Q &\to T^*Q \\
        (q_j,q_{j+1}) &\mapsto (q_{j}, p_{j}). \label{definition_gamma_minus}
\end{aligned}
\end{equation}
% \begin{flalign}
%    \gamma^+ = DS_d\circ \pi_2 + f_d^+: Q\times Q &\to T^*Q \nonumber\\
%         (q_j,q_{j+1}) &\mapsto (q_{j+1}, p_{j+1}), \label{definition_gamma_plus}\\
%    \gamma^- = DS_d\circ \pi_1 - f_d^-: Q\times Q &\to T^*Q \nonumber\\
%         (q_j,q_{j+1}) &\mapsto (q_{j}, p_{j}) \label{definition_gamma_minus},
% \end{flalign}
\label{definitions_gamma_plus_minus}
\end{subequations}
% where $\pi_1,\pi_2:Q\times Q\to Q$ are the natural projections given by $\pi_1(q_i,q_j)=q_i$ and $\pi_2(q_i,q_j)=q_j$. 
% \todo{\tiny Do I need this assumption?}
% Assume that $Q$ is a an open subset of $\RR^n$.
Consider the bundle isomorphism
\begin{equation}
\begin{aligned}
  \Phi: T^*(Q\times Q)&\to T^*Q\times T^*Q\\
  (q_i,q_j,p_i,p_j)&\mapsto (q_i,p_i,q_j,p_j).
\end{aligned}
\end{equation}
Let us define the discrete section $\gamma$ given by
\begin{equation}
\begin{aligned} 
  \gamma:Q\times Q&\to T^*(Q\times Q)  \\
  (q_j,q_{j+1})
 &\mapsto \Phi^{-1}\left( \gamma^-(q_j,q_{j+1}),\gamma^+(q_j,q_{j+1})  \right)=
 (q_j, q_{j+1}, p_j, p_{j+1}).
\end{aligned}
\end{equation}
% Then the discrete Hamiltonian flow can be written as
% \begin{equation}
% \mathcal{F}_{d}^{H}:
% \left( q_j, \gamma(q_{j-1},q_j)  \right)
% \mapsto
% \left( q_{j+1}, \gamma(q_{j},q_{j+1})  \right).
% % \left(q_j, DS_d^j(q_j)+ f_d^+(q_{j-1},q_{j})   \right)
% % \mapsto 
% % \left(q_{j+1}, DS_d^{j+1}(q_{j+1})  +f_d^+(q_j,q_{j+1})  \right).
% \end{equation}
We can now define the mappings 
\begin{equation}
\begin{aligned}
  \left(\mathcal{F}_d^H   \right)^{\gamma^+}
  =\pi_Q\circ \mathcal{F}_d^H \circ \gamma^+: Q\times Q&\to Q\\
  (q_{j-1},q_{j}) 
  &\mapsto  
  % \pi_Q \circ \mathcal{F}_d^H  (\gamma^+(q_{j},q_{j+1}))=
  q_{j+1},\\
   \left(\mathcal{F}_d^H   \right)^{\gamma^-}
  =\pi_Q\circ \mathcal{F}_d^H \circ \gamma^-: Q\times Q&\to Q\\
  (q_{j-1},q_{j}) 
  &\mapsto  q_{j}.
\end{aligned}
\end{equation}
Consider the flows given by
\begin{equation}
\begin{aligned} 
  \left(\mathcal{F}_d^H   \right)^{\gamma}
  % = \left( \left(\mathcal{F}_d^H   \right)^{\gamma^+},\left(\mathcal{F}_d^H   \right)^{\gamma^-}  \right)
  : Q\times Q&\to Q\times Q\\
    (q_{j-1},q_{j})  
    &\mapsto  \left( \left(\mathcal{F}_d^H   \right)^{\gamma^-} (q_{j-1},q_{j}),
      \left(\mathcal{F}_d^H   \right)^{\gamma^+}  (q_{j-1},q_{j})  \right)
     = (q_j,q_{j+1}),  
\end{aligned}
\end{equation}
and
\begin{equation}
\begin{aligned}
  \tilde{\mathcal{F}}_d^H
  % =\left( \mathcal{F}_d^H,\mathcal{F}_d^H  \right)
  :T^*(Q\times Q)&\to T^*(Q\times Q)\\
  \tilde{\mathcal{F}}_d^H(q_{j-1},q_j,p_{j-1},p_j)
  &= \Phi^{-1} \left( \mathcal{F}_d^H\circ \pi_- \circ \Phi (q_{j-1},q_j,p_{j-1},p_j),
     \mathcal{F}_d^H\circ \pi_+ \circ \Phi (q_{j-1},q_j,p_{j-1},p_j) \right)\\
  &=(q_j,q_{j+1},p_{j},p_{j+1}),
  % (q_{j-1},q_j,p_{j-1},p_j)
  % &\mapsto \Phi^{-1} \left( \mathcal{F}_d^H\circ \pi_- \circ \Phi (q_{j-1},q_j,p_{j-1},p_j),
  %    \mathcal{F}_d^H\circ \pi_+ \circ \Phi (q_{j-1},q_j,p_{j-1},p_j) \right)\\
  % &=(q_j,q_{j+1},p_{j},p_{j+1}),
\end{aligned}
\end{equation}
where $\pi_-,\pi_+:T^*Q\times T^*Q\to T^*Q$ are the projections on the first and on the second factor of $T^*Q$, respectively.
It can be easily checked that
\begin{equation}
  \left( \mathcal{F}_d^H  \right)^\gamma
  =\pi_{Q\times Q}\circ \tilde{\mathcal{F}}_d^H \circ \gamma,
  \label{eq_commutative_diagram_flows}
\end{equation}
or, in other words, the following diagram commutes:
\begin{center}
% https://tikzcd.yichuanshen.de/#N4Igdg9gJgpgziAXAbVABwnAlgFyxMJZABgBpiBdUkANwEMAbAVxiRABUA9AKgAoBFADqC8AW3gACfgEoQAX1LpMufIRQBmclVqMWbLnyEis4uFNkKl2PASJkATNvrNWiEEbGT+8xSAzXVIk1Hamc9Nw8TL3ltGCgAc3giUAAzACcIUSR7ahwIJDIdFzZheLpRUToQagY6ACMYBgAFZRs1EDSseIALHB9UjKzEQrykAEZQ3VcQYTwGWGBhSpxugGNGYAAxOTkAfShOAAl+kHTM7Nz8xE0i8JBeJboV9YYtvYPD6U5S8sqTs6GExAo2uk2KbmEaCwu2AkVMUgkcmqIAaYCgSHUhTg3SwKT6iAAtGNLKdBgVLhcUTA0UgCZiwXdIdDYbMomZ+EiavVGi0ArY3J0en0SQCMRTEECwtMfhUqnIKHIgA
\begin{tikzcd}
T^*(Q\times Q) \arrow[rrr, "\tilde{\mathcal{F}}_d^H"] \arrow[dd, "\pi_{Q\times Q}"', bend right] &  &  & T^*(Q\times Q) \arrow[dd, "\pi_{Q\times Q }", bend left] \\
                                                                                                 &  &  &                                                                      \\
Q\times Q \arrow[uu, "\gamma"'] \arrow[rrr, "(\mathcal{F}_d^H)^\gamma"]                          &  &  & Q\times Q \arrow[uu, "\gamma"]                                      
\end{tikzcd}
\end{center}
The pointwise interpretation of the diagram above is

\begin{center}
% https://tikzcd.yichuanshen.de/#N4Igdg9gJgpgziAXAbVABwnAlgFyxMJZABgBpiBdUkANwEMAbAVxiRAAoBHAfWACsAtAEYAvqR59SaXoNFTufAJQgx6TLnyEUAZnJVajFmy4LxMgNRzpk6f0sjlqkBmx4CRMgCZ99Zq0QcPPzCYhKOpGqumkS63tS+RgEm-KEWoo76MFAA5vBEoABmAE4QALZIntQ4EEhkBn5sADqN2XSlpXQg1Ax0AEYwDAAK6m5aIEVY2QAWOCoRIMVltVU1iELxhv4gzXgMsMDNHThTAMaMwABiIiLcUAB6ABJzhSXliJUg1Ui69YnbjUdTucrrc7s1Wu1Ok5Fm91p9Vj8EltmmgsLwAIo7LCleAAAnRuJEXRA-TAUG+dTgUywBVmiGI0Ney3hFWopPJiAE2jqSKajVRGKxOLg+KJ3T6A2GUXcAQm01mjKWiB+XzWGwaAXBbQ6KgoIiAA
\begin{tikzcd}
{(q_{j-1},q_j,p_{j-1},p_j)} \arrow[rrr, "\tilde{\mathcal{F}}_d^H"] \arrow[dd, "\pi_{Q\times Q}"', bend right] &  &  & {(q_j,q_{j+1},p_j,p_{j+1})} \arrow[dd, "\pi_{Q\times Q }", bend left] \\
                                                                                                              &  &  &                                                                       \\
{(q_{j-1},q_j)} \arrow[uu, "\gamma"'] \arrow[rrr, "(\mathcal{F}_d^H)^\gamma"]                                     &  &  & {(q_{j},q_{j+1})} \arrow[uu, "\gamma"]                               
\end{tikzcd}

\end{center}

\begin{proposition}
% \todo[inline]{Define $(\mathcal{F}_d^{H})^{DS_d}$}
The flows $(\mathcal{F}_d^{H})^{\gamma}$ and $\tilde{\mathcal{F}}_d^{H}$ are $\gamma$-related if the equations
\begin{subequations}
\begin{flalign}
  & S_{d}^{j+1}\left(q_{j+1}\right)-S_{d}^{j}\left(q_{j}\right)
  -\gamma^+ (q_j,q_{j+1}) q_{j+1}
  +H_{d}^{+}\left(q_{j}, \gamma^+ (q_j,q_{j+1})\right)=0, \label{forced_right_discrete_H-J}
  \\
  & S_{d}^{j+1}\left(q_{j+1}\right)-S_{d}^{j}\left(q_{j}\right)
  +\gamma^- (q_j,q_{j+1}) q_{j}
  +H_{d}^{-}\left(q_{j+1}, \gamma^- (q_j,q_{j+1})\right)=0, \label{forced_left_discrete_H-J}
\end{flalign}
\label{forced_discrete_H-J}
\end{subequations}
are satisfied. We shall call these equations the forced right discrete Hamilton-Jacobi equation (FRDHJ) and the forced left discrete Hamilton-Jacobi equation (FLDHJ), respectively.
\end{proposition}

\begin{proof}
Eq.~\eqref{forced_right_discrete_H-J} follows immediately from Eqs.~\eqref{derive_H-J_1} and \eqref{derive_H-J_2}. Similarly, we can show Eq.~\eqref{forced_left_discrete_H-J}.

On the other hand, Eq.~\eqref{eq_commutative_diagram_flows} is equivalent to saying that, at any point, 
\begin{equation}
  \tilde{\mathcal{F}}_d^H \circ \gamma\left( q_{j-1},q_j  \right)
  % =\left( (\mathcal{F}_d^H)^\gamma\left( q_{j-1},q_j  \right), \gamma \circ (\mathcal{F}_d^H)^\gamma \left( q_{j-1},q_j  \right)   \right)
  =\left( q_{j},q_{j+1}, \gamma\left( q_{j},q_{j+1}  \right)  \right).
\end{equation}
We have that
\begin{equation}\begin{aligned}
  \tilde{\mathcal{F}}_d^H \circ \gamma\left( q_{j-1},q_j  \right)
  &=\tilde{\mathcal{F}}_d^H \left( q_{j-1},q_j, p_{j-1},p_j  \right)\\
  &=\Phi^{-1} \left( {\mathcal{F}}_d^H(q_{j-1},p_{j-1}),
                    {\mathcal{F}}_d^H(q_{j},p_{j})  \right)\\
  &=\Phi^{-1} \left( \left( q_{j}, D_2 L_d (q_{j-1}, q_{j}) + f_d^+ (q_{j-1}, q_{j})  \right),\right.\\ &\qquad \quad \left.
                    \left( q_{j+1}, D_2 L_d (q_{j}, q_{j+1}) + f_d^+ (q_{j}, q_{j+1})  \right)
              \right)\\
  &= \left( q_{j}, q_{j+1}, D_2 L_d (q_{j-1}, q_{j}) + f_d^+ (q_{j-1}, q_{j}),
          D_2 L_d (q_{j}, q_{j+1}) + f_d^+ (q_{j}, q_{j+1}) \right),
\end{aligned}\end{equation}
% and
% \begin{equation}
% \begin{aligned} 
%   \gamma \circ (\mathcal{F}_d^H)^\gamma (q_{j-1},q_j) 
%   = \left((q_{j}, q_{j+1}),   \gamma (q_{j}, q_{j+1})  \right).
% \end{aligned}
% \end{equation}
% The commutativity of the diagram implies that
The diagram above commutes if and only if
\begin{equation}
\begin{aligned} 
  \gamma (q_{j}, q_{j+1}) = \left( D_2 L_d (q_{j-1}, q_{j}) + f_d^+ (q_{j-1}, q_{j}),
          D_2 L_d (q_{j}, q_{j+1}) + f_d^+ (q_{j}, q_{j+1})  \right),
\end{aligned}
\end{equation}
that is,
\begin{equation}
\begin{aligned} 
   & \gamma^- (q_{j}, q_{j+1}) = D_2 L_d (q_{j-1}, q_{j}) + f_d^+ (q_{j-1}, q_{j}),\\
   & \gamma^+ (q_{j}, q_{j+1}) = D_2 L_d (q_{j}, q_{j+1}) + f_d^+ (q_{j}, q_{j+1}).
\end{aligned}
\end{equation}
By Eq.~\eqref{forced_right_discrete_H-J}, this is equivalent to
\begin{equation}
   \left[ D_2 L_d (q_{j}, q_{j+1}) + f_d^+ (q_{j}, q_{j+1}) \right] q_{j+1}
  = H_{d}^{+}\left(q_{j}, p_{j+1}\right) + S_{d}^{j+1}\left(q_{j+1}\right)-S_{d}^{j}\left(q_{j}\right),
\end{equation}
where 
\begin{equation}
  p_{j+1} = \gamma^+ (q_j,q_{j+1}),
\end{equation}
which is true by using definition \eqref{definition_gamma_plus}. Similarly, by Eq.~\eqref{forced_left_discrete_H-J}, we have that
\begin{equation}
  \left[ D_2 L_d (q_{j-1}, q_{j}) + f_d^+ (q_{j-1}, q_{j})  \right]
  q_j = S_d^j(q_j) - S_d^{j+1}(q_{j+1}) - H_d^-(q_{j+1}, p_j),
\end{equation}
where 
\begin{equation}
  p_{j} = \gamma^- (q_j,q_{j+1}),
\end{equation}
which is true by using definition \eqref{definition_gamma_minus}.

% Clearly,
% \begin{equation}
%   \mathcal{F}_d^H \circ \gamma^+(q_{j-1},q_j)
%   = \mathcal{F}_d^H(q_j,p_j)
%   = (q_{j+1},p_{j+1}),
% \end{equation}
% and
% \begin{equation}
%   \gamma^+ \left(q_j,
%               (\mathcal{F}_d^{H})^{DS_d}(q_{j-1},q_{j})    \right)
%   = \gamma^+(q_{j},q_{j+1})=(q_{j+1},p_{j+1}).
% \end{equation}

\end{proof}

\newcommand*{\flowp}{\mathcal{F}^+}
\newcommand*{\flowm}{\mathcal{F}^-}

\begin{lemma}\label{lemma_discrete_HJ}
% [Relation FRDHJ and Hamilton's eqs]
Suppose that $S_d^k$ and $\gamma^+$ satisfy the FRDHJ \eqref{forced_right_discrete_H-J}, and let $\left\{c_k  \right\}_{k=0}^N\subset Q$ be a sequence of points such that
\begin{equation}
  % (c_k,c_{k+1}) = (\mathcal{F}_d^H)^\gamma(c_{k-1}, c_k).
  c_{k+1} = \flowp (c_{k-1}, c_k), 
  \label{c_k_hat_flowp}
\end{equation}
where $\flowp:Q\times Q\to Q$ is implicitly defined by
\begin{equation}
\begin{aligned}
 \flowp(q_{k-1},q_k) =& D_2 H_d^+ \left( q_k, \gamma^+(q_k, \flowp(q_{k-1},q_k))  \right)
 -f_d^+ \left( q_k,  \flowp(q_{k-1},q_k) \right) \\
 &\times\left[ D_2\gamma^+(q_k, \flowp(q_{k-1},q_k))   \right]^{-1}
 \label{implicit_definition_flowp}
\end{aligned}
\end{equation}
Then, the sequence of points $\left\{(c_k,p_k)  \right\}_{k=0}^N\subset T^*Q$ with
\begin{equation}
  p_{k+1} = \gamma^+(c_{k},c_{k+1})
  % p_k = \gamma^- (c_k,c_{k+1}) \label{momentum_gamma}
\end{equation}
is a solution of the forced right discrete Hamilton equations \eqref{right_discrete_Hamilton}.
%%%%%%%%%%%%

Similarly, suppose that $S_d^k$ and $\gamma^-$ satisfy the FLDHJ \eqref{forced_left_discrete_H-J}, and let $\left\{c_k  \right\}_{k=0}^N\subset Q$ be a sequence of points such that
\begin{equation}
  % (c_k,c_{k+1}) = (\mathcal{F}_d^H)^\gamma(c_{k-1}, c_k).
  c_{k+1} = \flowm (c_{k-1}, c_k), 
  \label{c_k_hat_flowm}
\end{equation}
where $\flowm:Q\times Q\to Q$ is implicitly defined by
\begin{equation}
\begin{aligned}
 q_k &= 
 f_d^- \left( q_k,  \flowm(q_{k-1},q_k) \right) 
 \left[ D_{q_k} \gamma^-(q_k, \flowm(q_{k-1},q_k))   \right]^{-1}\\
 &\ - D_2 H_d^- \left( \flowm(q_{k-1},q_k), \gamma^-(q_k, \flowm(q_{k-1},q_k))  \right)
 \label{implicit_definition_flowm}
\end{aligned}
\end{equation}
Then, the sequence of points $\left\{(c_k,p_k)  \right\}_{k=0}^N\subset T^*Q$ with
\begin{equation}
  % p_{k} = \gamma^-(c_{k-1},c_k)
  (p_k, p_{k+1}) = \gamma(c_{k},c_{k+1})
\end{equation}
is a solution of the forced left discrete Hamilton equations \eqref{left_discrete_Hamilton}.
% Moreover, $\flowm=\left(\mathcal{F}_d^H  \right)^{\gamma^+}$.
\end{lemma}

\begin{proof}
Replacing $q_{k+1}$ by $\flowp(q_{k-1},q_k)$ in Eq.~\eqref{forced_right_discrete_H-J}, we can write
\begin{equation}
\begin{aligned} 
  S_{d}^{k+1}\left(\flowp(q_{k-1}, q_k)\right)
  &-S_{d}^{k}\left(q_{k}\right)
  -\gamma^+ (q_k,\flowp(q_{k-1}, q_k)) \flowp(q_{k-1}, q_k)\\
  &+H_{d}^{+}\left(q_{k}, \gamma^+ (q_k,\flowp(q_{k-1}, q_k))\right)=0, 
\end{aligned}
\end{equation}
which derived with respect to $q_k$ yields
\begin{equation}
\begin{aligned}
  DS_d^{k+1} &\left( \flowp(q_{k-1}, q_k)  \right)  D_2 \flowp(q_{k-1}, q_k) 
  -DS_d^k (q_k)
  -D_{q_k} \gamma^+ \left( q_k, \flowp(q_{k-1}, q_k)   \right) \flowp(q_{k-1}, q_k) \\
  &- \gamma^+ \left( q_k, \flowp(q_{k-1}, q_k)   \right) D_2 \flowp(q_{k-1}, q_k) 
  + D_1 H_d^+ \left( q_k, \gamma^+(q_k, \flowp(q_{k-1}, q_k) )  \right)\\
  &+ D_2 H_d^+ \left( q_k, \gamma^+(q_k, \flowp(q_{k-1}, q_k) )  \right)
      D_{q_k} \gamma^+ \left( q_k, \flowp(q_{k-1}, q_k)   \right) = 0,
\end{aligned}
\end{equation}
which, due to Eqs.~\eqref{definition_gamma_plus} and \eqref{implicit_definition_flowp}, reduces to
\begin{equation}
\begin{aligned}
  -f_d^+ & \left(q_k, \flowp(q_{k-1}, q_k)  \right)  D_2 \flowp(q_{k-1}, q_k)
  -DS_d^k (q_k)\\
 & +f_d^+ \left(q_k, \flowp(q_{k-1}, q_k)  \right) D_{q_k} \gamma^+ \left( q_k, \flowp(q_{k-1}, q_k)   \right) \left[ D_2\gamma^+(q_k, \flowp(q_{k-1},q_k))   \right]^{-1}\\
 & + D_1 H_d^+ \left( q_k, \gamma^+(q_k, \flowp(q_{k-1}, q_k) )  \right) = 0.
\end{aligned}
\end{equation}
In particular, for $q_k=c_k$, we have
\begin{equation}
\begin{aligned}
  -f_d^+ &\left(c_k, \flowp(c_{k-1}, c_k)  \right)  D_2 \flowp(c_{k-1}, c_k)
  -DS_d^k (c_k)\\
 & +f_d^+ \left(c_k, \flowp(c_{k-1}, c_k)  \right) D_{c_k} \gamma^+ \left( c_k, \flowp(c_{k-1}, c_k)   \right) \left[ D_2\gamma^+(c_k, \flowp(c_{k-1},c_k))   \right]^{-1}\\
 & + D_1 H_d^+ \left( c_k, \gamma^+(c_k, \flowp(c_{k-1}, c_k) )  \right) = 0.
\end{aligned}
\end{equation}
In other words,
\begin{equation}
    -f_d^+ \left(c_k, c_{k+1}  \right)   \frac{\partial c_{k+1}} {\partial  c_{k}}
  -DS_d^k (c_k)
  +f_d^+ \left(c_k, c_{k+1}  \right) \frac{\partial p_{k+1}} {\partial  c_k} \left( \frac{\partial p_{k+1}} {\partial  c_{k+1}}  \right)^{-1}
  + D_1 H_d^+ \left( c_k, p_{k+1}  \right) = 0,  
\end{equation}
that is,
\begin{equation}
   D_1 H_d^+ \left( c_k, p_{k+1}  \right) -DS_d^k (c_k) = 0.
\end{equation}
By Eq.~\eqref{definition_gamma_minus}, we have that
\begin{equation}
  D_1 H_d^+ \left( c_k, p_{k+1}  \right) - p_k - f_d^-(c_k, c_{k+1}) = 0.
  \label{eq_Hamilton_proof_HJ_theorem}
\end{equation}
This last equation, together with the definition of $\flowp$, show that the sequence $\left\{(c_k,p_k)  \right\}$ satisfies the forced right discrete Hamilton equations.

We shall now prove the second assertion. The FLDHJ \eqref{forced_left_discrete_H-J} can be written as
\begin{equation}
\begin{aligned}
    S_d^{k+1} \left(\flowm(q_{k-1},q_k)  \right)& - S_d^k(q_k)
    + \gamma^- \left(q_k, \flowm(q_{k-1},q_k)  \right) q_k\\
    &+ H_d^- \left(\flowm(q_{k-1},q_k), \gamma^- \left(q_k, \flowm(q_{k-1},q_k)  \right)  \right) = 0,
\end{aligned}
\end{equation}
whose derivative with respect to $q_k$ is
\begin{equation}
\begin{aligned}
    DS_d^{k+1} & \left(\flowm(q_{k-1},q_k)  \right) D_2\flowm(q_{k-1},q_k)
    - DS_d^k(q_k)
    + D_{q_k} \gamma^- \left(q_k, \flowm(q_{k-1},q_k)  \right) q_k\\
    &+ \gamma^- \left(q_k, \flowm(q_{k-1},q_k)  \right) 
    + D_1 H_d^- \left(\flowm(q_{k-1},q_k), \gamma^- \left(q_k, \flowm(q_{k-1},q_k)  \right)  \right) 
        D_2 \flowm(q_{k-1},q_k)\\
    &+ D_2 H_d^- \left(\flowm(q_{k-1},q_k), \gamma^- \left(q_k, \flowm(q_{k-1},q_k)  \right)  \right) 
        D_{q_k} \gamma^- \left(q_k, \flowm(q_{k-1},q_k) \right)
    = 0.
\end{aligned}
\end{equation}
By Eqs.~\eqref{definitions_gamma_plus_minus} and \eqref{implicit_definition_flowm}, the last equation reduces to
\begin{equation}
    \left[
        DS_d^{k+1} \left(\flowm(q_{k-1},q_k)  \right) 
        + D_1 H_d^- \left(\flowm(q_{k-1},q_k), \gamma^- \left(q_k, \flowm(q_{k-1},q_k)  \right)  \right) 
    \right] D_2\flowm(q_{k-1},q_k) = 0,
\end{equation}
so
\begin{equation}
    \gamma^+ \left(q_k, \flowm(q_{k-1},q_k)   \right)
    - f_d^+ \left(q_k, \flowm(q_{k-1},q_k)   \right)
    + D_1 H_d^- \left(\flowm(q_{k-1},q_k), \gamma^- \left(q_k, \flowm(q_{k-1},q_k)  \right) \right)
    =0.
\end{equation}
In particular, for $q_k=c_k$, we have 
\begin{equation}
    \gamma^+ \left(c_k, \flowm(c_{k-1},c_k)   \right)
    - f_d^+ \left(c_k, \flowm(c_{k-1},c_k)   \right)
    + D_1 H_d^- \left(\flowm(c_{k-1},c_k), \gamma^- \left(c_k, \flowm(c_{k-1},c_k)  \right) \right)
    =0,
\end{equation}
that is,
\begin{equation}
    p_{k+1} - f_d^+(c_k, c_{k+1}) + D_1 H_d^- (c_{k+1}, p_k)=0,
\end{equation}
which, together with the definition of $\flowm$ prove that the forced left discrete Hamilton equations are satisfied. 
\end{proof}

\begin{theorem}\label{theorem_discrete_HJ}

%%%%%%%%%%%%%%%%%%%

Suppose that $S_d^k$ and $\gamma^-$ satisfy the FLDHJ \eqref{forced_left_discrete_H-J}, and let $\left\{c_k  \right\}_{k=0}^N\subset Q$ be a sequence of points such that
\begin{equation}
  % (c_k,c_{k+1}) = (\mathcal{F}_d^H)^\gamma(c_{k-1}, c_k).
  c_{k+1} = \left(\mathcal{F}_d^H  \right)^{\gamma^+} (c_{k-1}, c_k), 
  % \label{c_k_hat_flowm}
\end{equation}
Then, the sequence of points $\left\{(c_k,p_k)  \right\}_{k=0}^N\subset T^*Q$ with
\begin{equation}
  % p_{k} = \gamma^-(c_{k-1},c_k)
  (p_k, p_{k+1}) = \gamma(c_{k},c_{k+1})
\end{equation}
is a solution of the forced left discrete Hamilton equations \eqref{left_discrete_Hamilton}.

%%%%%%%%%%%%%
Let us define the mapping
\begin{equation}
    \mathcal{F}^{\leftarrow} = \pi_Q \circ \Legm \circ \left(\Legp  \right)^{-1} \circ \gamma^-: (q_{j}, q_{j+1}) \mapsto q_{j-1},
\end{equation}
which goes backwards in time. Then, if $S_d^k$ and $\gamma^+$ satisfy the FRDHJ, and $\left\{c_k  \right\}_{k=0}^N\subset Q$ is given by
\begin{equation}
    c_{k-1} = \mathcal{F}^{\leftarrow} (c_k, c_{k+1}),
\end{equation}
then the sequence $\left\{(c_k,p_k)  \right\}_{k=0}^N\subset T^*Q$ with
\begin{equation}
   (p_k, p_{k+1}) = \gamma(c_{k},c_{k+1})
\end{equation}
is a solution of the right discrete Hamilton equations \eqref{right_discrete_Hamilton}.
\end{theorem}

\begin{proof}%[Proof of Theorem \ref{theorem_discrete_HJ}]

Observe that
\begin{equation}
    D_1 L_d(q_j, q_{j+1}) = -p_j - \left[q_j + D_2 H_d^-(q_{j+1}, p_j) \right] \frac{\partial p_j} {\partial q_j},
\end{equation}
so $\Legm(q_j,q_{j+1})=(q_j, p_{j})$ 
if and only if Eq.~\eqref{left_discrete_Hamilton_a} is satisfied. 
 If that is the case, we deduce that
\begin{equation}
    % \left(\mathcal{F}_d^H  \right)^{\gamma^+} (q_j, q_{j+1})
     \pi_Q \circ \mathcal{F}_d^H (q_j, p_j)
     = \pi_Q \circ \Legp \circ \left(\Legm  \right)^{-1} (q_j, p_j) 
     = q_{j+1}.
\end{equation}
If we now take $(p_j, p_{j+1})=\gamma(q_{j}, q_{j+1})$, then Eqs.~\eqref{right_discrete_Hamilton_a} and \eqref{left_discrete_Hamilton_a} are satisfied by construction,
 and we obtain 
\begin{equation}
   \left(\mathcal{F}_d^H  \right)^{\gamma^+} (q_{j-1}, q_{j}) = q_{j+1}.
\end{equation}
Clearly, Eq.~\eqref{implicit_definition_flowm} holds for $\left(\mathcal{F}_d^H  \right)^{\gamma^+}=\flowm$.
% % Hence,
% \begin{equation}
% \begin{aligned}
%   & f_d^- \left( q_k,  \flowm(q_{k-1},q_k) \right) 
%   \left[ D_{q_k} \gamma^-(q_k, \flowm(q_{k-1},q_k))   \right]^{-1}\\
%   &\quad - D_2 H_d^- \left( \flowm(q_{k-1},q_k), \gamma^-(q_k, \flowm(q_{k-1},q_k))  \right) 
%   % =  f_d^- \left( q_k,  q_{k+1} \right) 
%   % \left[ D_{q_k} \gamma^-(q_k, q_{k+1})   \right]^{-1}
%   % - D_2 H_d^- \left( q_{k+1}, \gamma^-(q_k, q_{k+1})  \right)  \\
%   % =  f_d^- \left( q_k,  q_{k+1} \right) 
%   % \left[ D_{q_k} p_k  \right]^{-1}
%   % - D_2 H_d^- \left( q_{k+1}, p_k \right)  
%   = q_k.
% \end{aligned}
% \end{equation}
The proof of the converse is analogous.
\end{proof}

\subsection{Hamilton-Jacobi theory for discrete Rayleigh systems}
% \begin{remark}[Hamilton-Jacobi theory for discrete Rayleigh systems]

If the discrete force is Rayleigh, we have that
\begin{equation}
\begin{aligned}
    &\gamma^- = DS_d\circ \pi_1 - D_1 \mathcal{R}_d,
    &\gamma^+ = DS_d\circ \pi_2 - D_2 \mathcal{R}_d.
\end{aligned}
\end{equation}
Let us introduce the function
\begin{equation}
\begin{aligned}
    \tilde{S}_d:Q\times Q &\to \RR\\
    (q_i, q_j) & \mapsto S^i(q_i) + S^j(q_j),
\end{aligned}
\end{equation}
so that
\begin{equation}
    \gamma = D \tilde{S}_d - D \mathcal{R}_d.
\end{equation}
Moreover, we can define $G_d= \tilde{S}_d - \mathcal{R}_d$, and write
\begin{equation}
    \gamma = D G_d.
\end{equation}
Then
\begin{equation}
\begin{aligned}
    \gamma^- = D_1 G_d,\\
    \gamma^+ = D_2 G_d,  
\end{aligned}
\end{equation}
so the FRDHJ and FLDHJ can be written as
\begin{subequations}
\begin{equation}
    S_{d}^{j+1}\left(q_{j+1}\right)-S_{d}^{j}\left(q_{j}\right)
  -D_2 G_d (q_j,q_{j+1}) q_{j+1}
  +H_{d}^{+}\left(q_{j}, D_2 G_d (q_j,q_{j+1})\right)=0,
\end{equation}
and
\begin{equation}
    S_{d}^{j+1}\left(q_{j+1}\right)-S_{d}^{j}\left(q_{j}\right)
  +D_1 G_d (q_j,q_{j+1}) q_{j}
  +H_{d}^{-}\left(q_{j+1}, D_1 G_d (q_j,q_{j+1})\right)=0, \label{forced_left_discrete_H-J_Rayleigh}
\end{equation}
\end{subequations}
respectively.

The second statement of Theorem \ref{theorem_discrete_HJ} can be particularized for discrete Rayleigh systems as follows.

\begin{corollary}
% Assume that the first of the left discrete Hamilton equations,
% \begin{equation}
%    \left[  q_{j} +D_{2} H_{d}^{-}\left(q_{j+1}, D_1 G_d (q_j, q_{j+1})\right)  \right] 
%   D_1^2 \tilde{L}_d (q_j, q_{j+1})
%    =  D_1 \tilde{R}_d (q_j, q_{j+1}) ,
%    \label{LDH_1_Rayleigh}
% \end{equation}
% holds, and let
Let
\begin{equation}
    \flowm = \pi_Q \circ \mathcal{F}_d^H \circ D_2 G_d.
\end{equation}
Suppose that $S_d$ and $G_d$ satisfy the FLDHJ \eqref{forced_left_discrete_H-J_Rayleigh}, and let $\left\{c_k  \right\}_{k=0}^N\subset Q$ be a sequence of points such that
\begin{equation}
  % (c_k,c_{k+1}) = (\mathcal{F}_d^H)^\gamma(c_{k-1}, c_k).
  c_{k+1} = \flowm (c_{k-1}, c_k).
\end{equation}
% and 
% \begin{equation}
%    \left[  c_{j} +D_{2} H_{d}^{-}\left(c_{j+1}, D_1 G_d(c_j, c_{j+1})\right)  \right] 
%   D_1^2 G_d (c_j, c_{j+1})
%    =  D_1 \mathcal{R}_d (c_j, c_{j+1}) .
%    \label{LDH_1_Rayleigh}
% \end{equation}
Then, the sequence of points $\left\{(c_k,p_k)  \right\}_{k=0}^N\subset T^*Q$ with
\begin{equation}
  % p_{k} = \gamma^-(c_{k-1},c_k)
  (p_k, p_{k+1}) = DG_d(c_{k},c_{k+1})
\end{equation}
is a solution of the forced left discrete Hamilton equations. 
\end{corollary}
% \end{remark}

\begin{example}

\begin{figure}[h!]
  \centering
  \begin{subfigure}[t]{.45\linewidth}
    \centering
    \includegraphics[width=\linewidth]{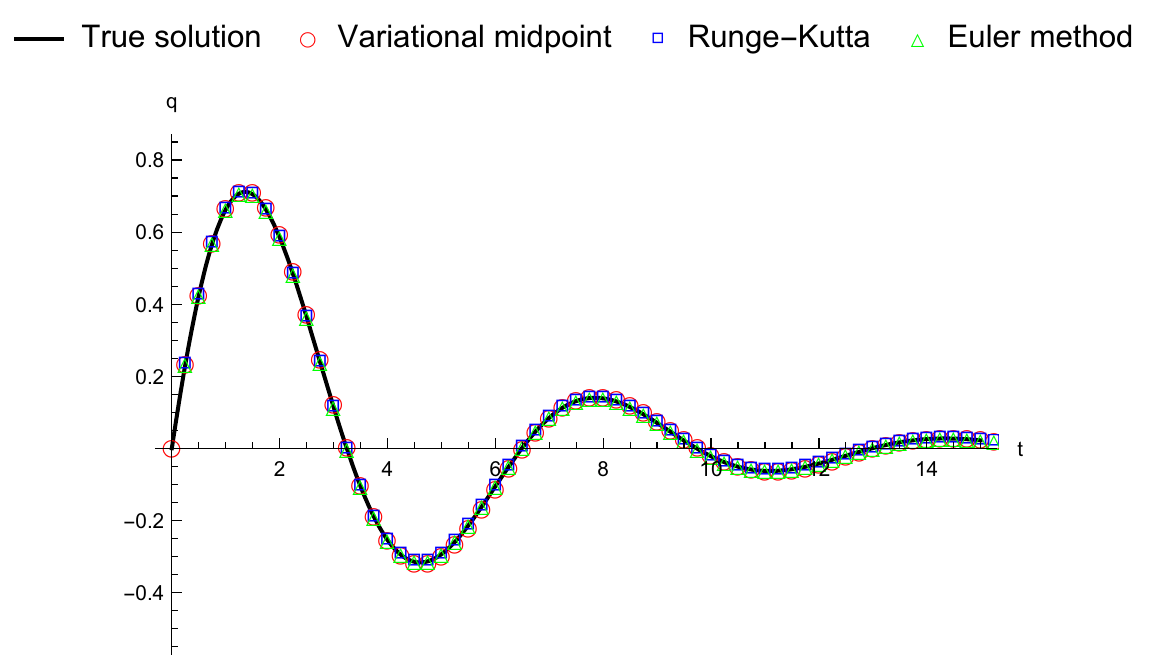}
    \caption{Trajectory}
    \label{plot_positions}
  \end{subfigure}
  \begin{subfigure}[t]{.45\linewidth}
    \centering
    \includegraphics[width=\linewidth]{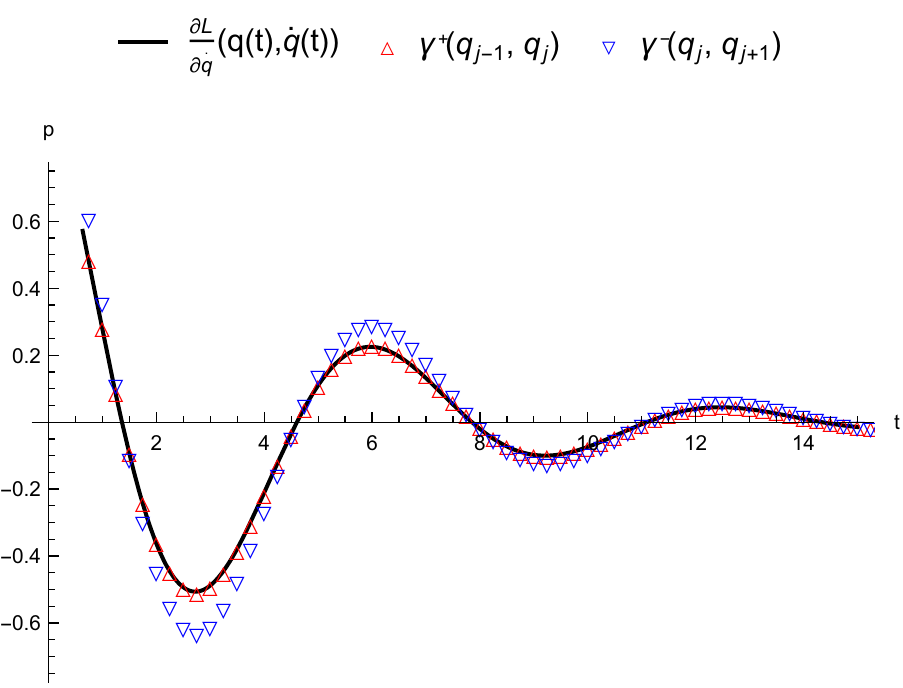}
    \caption{Momenta}
    \label{plot_momenta}
  \end{subfigure}
  \begin{subfigure}[t]{.45\linewidth}
    \centering
    \includegraphics[width=\linewidth]{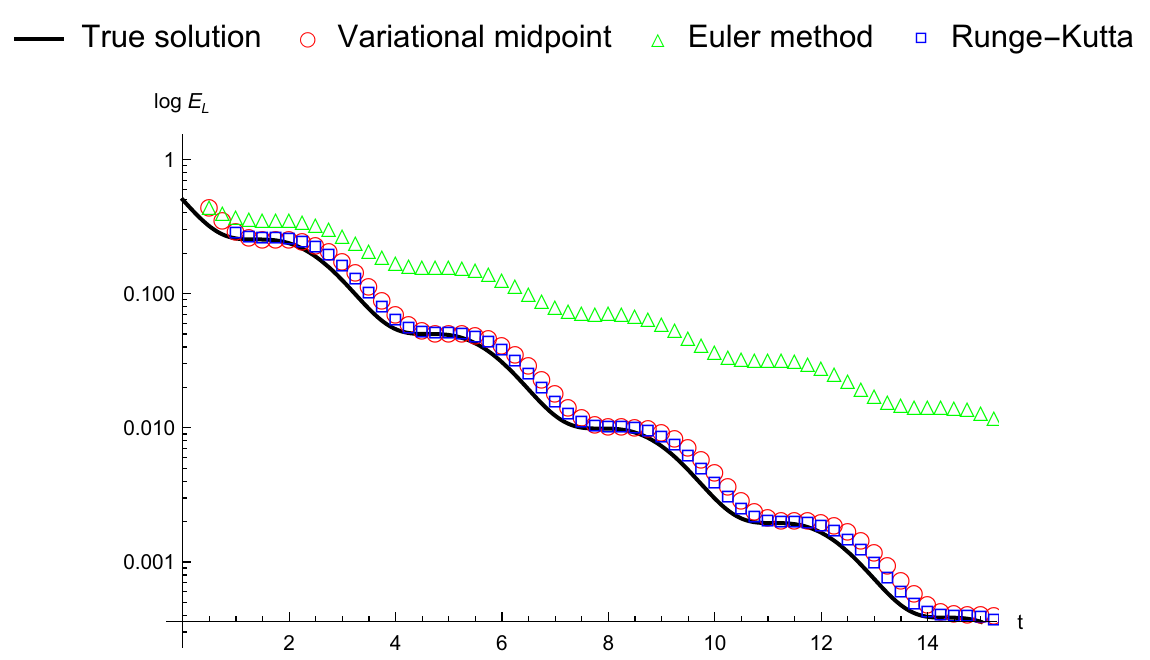}
    \caption{Energy}
    \label{plot_energy}
  \end{subfigure}
  \begin{subfigure}[t]{.45\linewidth}
    \centering
    \includegraphics[width=\linewidth]{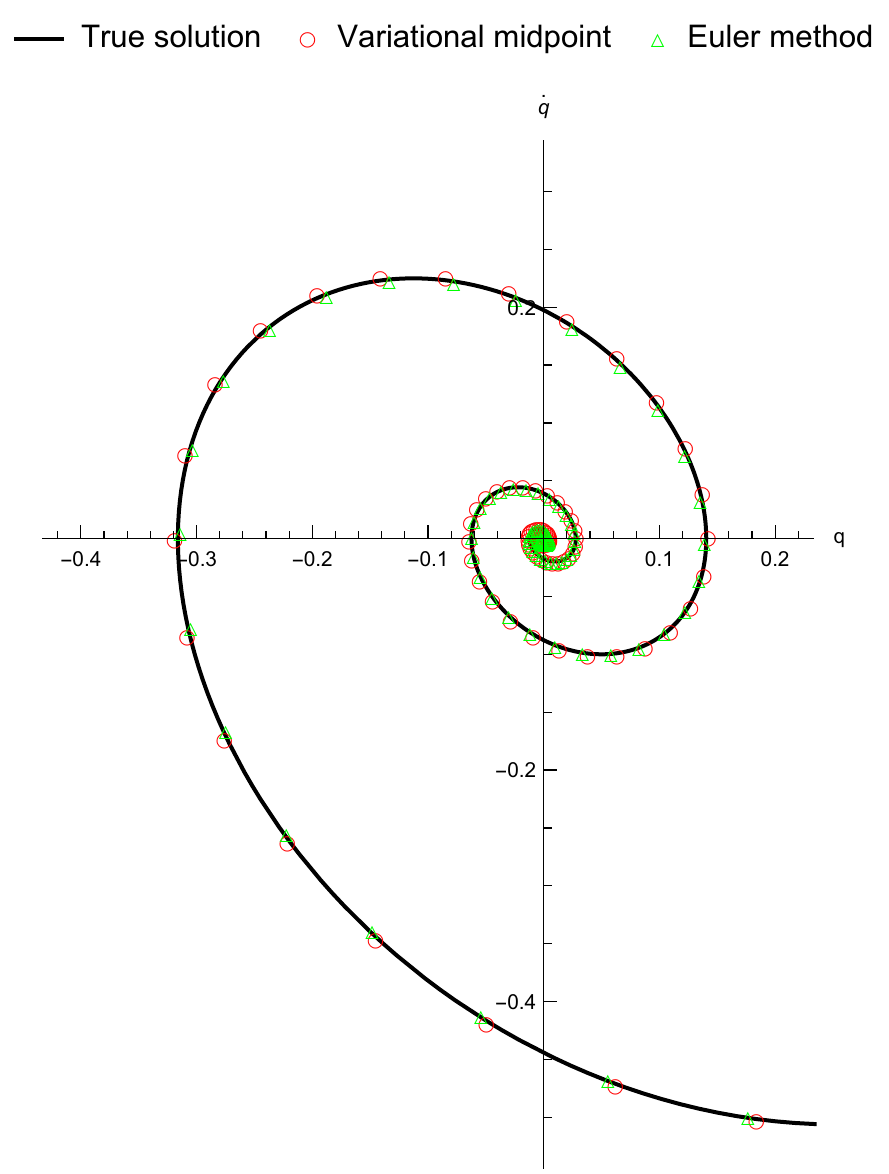}
    \caption{Space of velocities}
    \label{plot_phase_space}
  \end{subfigure}
  \caption{Harmonic oscillator subject to a linear Rayleigh dissipation.
  The true curves, given by the solution of the continuous forced Euler-Lagrange equation, are compared with the ones obtained from the variational midpoint rule and the Euler and Runge-Kutta methods.
   % The curves given by the solution of the continuous forced Euler-Lagrange equation are compared with the points given by the solutions of the forced discrete Euler-Lagrange equations, computed with the discrete Lagrangian and the discrete Rayleigh potential given by the midpoint rule approximation.
  }
  \label{plot_harmonic_oscillator}
\end{figure}

Consider the discrete Rayleigh system $(L_d, \mathcal{R}_d)$ on $\RR\times \RR$, where
\begin{equation}
  L_d (q_{j-1}, q_{j}) = \frac{h}{2} m \left(\frac{q_{j}-q_{j-1}}{h}  \right)^2 - \frac{h}{2}k \left(\frac{q_{j-1}+q_j}{2}  \right)^2
  ,
\end{equation}
and
\begin{equation}
  \mathcal{R}_d (q_{j-1}, q_{j}) =  r  \left(\frac{q_{j}-q_{j-1}}{2}  \right)^2.
\end{equation}
The corresponding discrete forces are
\begin{equation}
  f_d^+(q_{j-1}, q_{j}) = f_d^-(q_{j-1}, q_{j}) = -r \frac{q_{j}-q_{j-1}}{2}.
\end{equation}
% \todo[inline, color=red!40]{Correct equations! Interchange $L_d^+ \leftrightarrow L_d^-$}
% \todo[inline]{Faltaría comparar el orden del lagrangiano discreto y las fuerzas discretas}
Clearly, this system corresponds to the harmonic oscillator with a Rayleigh force (see Example \ref{example_harmonic_oscillator}) discretized via the midpoint rule. The modified discrete Lagrangians are 
\begin{equation}
  L_d^\pm (q_{j-1},q_{j}) = \frac{h}{2} m  \left(\frac{q_{j}-q_{j-1}}{h}  \right)^2 
  -\frac{h}{2} k  \left(\frac{q_{j}+q_{j-1}}{2}  \right)^2
  \pm  r  \left(\frac{q_{j}-q_{j-1}}{2}  \right)^2.
\end{equation}
The discrete Legendre transforms are given by
\begin{equation}
\begin{aligned}
  \Legp: (q_{j-1}, q_{j}) &\mapsto \left(q_{j}, \frac{m}{h}(q_{j}-q_{j-1}) - \frac{1}{4} kh (q_{j}+q_{j-1}) -r \frac{q_{j}-q_{j-1}}{2} \right),  \\
  \Legm: (q_{j-1}, q_{j}) &\mapsto \left(q_{j-1}, \frac{m-r}{h}(q_{j}-q_{j-1}) + \frac{1}{4} kh (q_{j}+q_{j-1}) + r \frac{q_{j}-q_{j-1}}{2} \right),
\end{aligned}
\end{equation}
and thus the discrete Hamiltonian flow is
\begin{equation}
  \mathcal{F}_{d}^{H}: (q_{j}, p_{j})  \mapsto 
  \left(\frac{h^2 (-k) q_j+4 h p_j+2 h q_j r+4 m q_j}{h^2 k+2 h r+4 m},-\frac{h^2 k p_j+4 h k m q_j+2 h p_j r-4 m p_j}{h^2 k+2 h r+4 m} \right).
\end{equation}
On the other hand, the discrete action is
\begin{equation}
  S_d^j (q_j) = \sum_{k=0}^{j-1} \left[\frac{h}{2} m \left(\frac{q_{k+1}-q_k}{h}  \right)^2
  - \frac{h}{2}k \left(\frac{q_{j+1}+q_{j}}{2}  \right)^2
  \right],
\end{equation}
so
% \begin{equation}
%   \tilde{S}_d(q_j, q_{j+1}) = \sum_{k=0}^{j-2} m \left(\frac{q_{k+1}-q_k}{h}  \right)^2
%   + m \left(\frac{q_{j}-q_{j-1}}{h}  \right)^2
%   +\frac{1}{2} m \left(\frac{q_{j+1}-q_j}{h}  \right)^2,
% \end{equation}
% and
\begin{equation}
\begin{aligned}
  G_d(q_j, q_{j+1}) 
 &= \sum_{k=0}^{j-2} \left[\frac{h}{2} m \left(\frac{q_{k+1}-q_k}{h}  \right)^2
  - \frac{h}{2}k \left(\frac{q_{j+1}+q_{j}}{2}  \right)^2
  \right]\\
  &+ m h \left(\frac{q_{j}-q_{j-1}}{h}  \right)^2
  - kh \left(\frac{q_{j}+q_{j-1}}{2}  \right)^2\\
 & +\frac{h}{2} m \left(\frac{q_{j+1}-q_j}{h}  \right)^2
  -\frac{h}{2} k \left(\frac{q_{j+1}+q_j}{2}  \right)^2
  -r  \left(\frac{q_{j+1}-q_{j}}{2}  \right)^2.
\end{aligned}
\end{equation}
Then
\begin{equation}
  \gamma^-   (q_j, q_{j+1}) 
  % = D_{1}G_d (q_j, q_{j+1}) 
  =  \left(\frac{m}{h} - \frac{3}{4} kh \right) q_j -  \left(\frac{m}{h} - \frac{1}{2} kh\right) q_{j+1} - \left(\frac{2m}{h} + \frac{kh}{2}\right) q_{j-1} +  r  \frac{q_{j+1}-q_{j}}{2} ,
\end{equation}
and
\begin{equation}
  % \gamma^+ (q_{j}, q_{j+1}) = D_{j+1} G_d(q_{j}, q_{j+1})  = (m+r) (q_{j+1}-q_{j}),
  \gamma^+ (q_j, q_{j+1}) 
  % = D_2 G_d (q_j, q_{j+1}) 
  = \frac{m}{h} (q_{j+1} - q_j)
  -\frac{1}{4}kh (q_j + q_{j+1})
  -r  \frac{q_{j+1}-q_{j}}{2}.
\end{equation}
 Therefore
\begin{equation}
  \flowm : (q_{j-1}, q_{j}) \mapsto
  \frac{h^2 (-k) (q_{j-1}+2 q_{j})+2 h q_{j-1} r-4 m (q_{j-1}-2 q_{j})}{h^2 k+2 h r+4 m}
   \label{example_flowm}
\end{equation}
The left discrete Hamiltonian is
\begin{equation}
  H_d^-(q_{j+1}, p_j) = -p_j q_j - \frac{h}{2} m \left(\frac{q_{j+1} - q_{j}}{h}  \right)^2 + \frac{h}{2}k \left(\frac{q_{j}+q_{j+1}}{2}  \right)^2.
\end{equation}
By construction, $S_d^j$ and $\gamma^-$ satisfy the FLDHJ \eqref{forced_left_discrete_H-J}. 
Then, the sequence of points $\left\{(c_k,p_k)  \right\}_{k=0}^N$, with $c_{j+1}=\flowm(c_{j-1}, c_j)$ and
\begin{equation}
  p_j = \gamma^-(c_j, c_{j+1}) 
  =  \left(\frac{m}{h} - \frac{3}{4} kh \right) c_j -  \left(\frac{m}{h} - \frac{1}{2} kh\right) c_{j+1} - \left(\frac{2m}{h} + \frac{kh}{2}\right) c_{j-1} +  r  \frac{c_{j+1}-c_{j}}{2} ,
\end{equation}
is a solution of the forced left discrete Hamilton equations for $(H_d^-, f_d)$. Observe that the points given by $c_{j+1}=\flowm(c_{j-1}, c_j)$ coincide with the solutions of the forced discrete Euler-Lagrange equations \eqref{discrete_forced_EL_Rayleigh}.

In Figure \ref{plot_harmonic_oscillator} we have plotted the position, momenta and energy of the system as a function of time, as well as the velocities as a function of the positions; comparing the solutions of the continuous forced Euler-Lagrange equations with the points given by the flow $\flowm$ (or, equivalently, the forced discrete Euler-Lagrange equations) as well as the Euler and fourth-order Runge-Kutta methods. However, this simple linear and 1-dimensional system does not show the advantages of variational methods over the standard numerical methods (see Example \ref{example_Marsden_West}).
\end{example}

\section{Conclusions and outlook}\label{section_conclusion}
In this paper we have defined a discrete analogue of the Rayleigh potential, $\mathcal{R}_d$. This has allowed us to write the Euler-Lagrange equations and the Legendre transformations of a forced discrete Rayleigh system $(L_d, \mathcal{R}_d)$ in terms of the modified discrete Lagrangians $L_d^\pm$ instead of the discrete Lagrangian $L_d$ and the discrete Rayleigh force $f_d$. We have studied the equivalence between discrete Rayleigh systems. Moreover, we have obtained a Noether's theorem for forced discrete Lagrangian systems, generalizing a result previously obtained by Marsden and West \cite{marsden_discrete_2001}; as well as other theorem which, given the Lie algebra of symmetries of the discrete Lagrangian, characterizes the Lie subalgebra of symmetries of the discrete force. Furthermore, we have developed a Hamilton-Jacobi theory for forced discrete Hamiltonian systems, based on the discrete flow approach.

This paper has continued our geometric study of mechanical systems with external forces, initiated with two previous papers. In the first paper we obtained a Noether's theorem and a symplectic reduction method for forced mechanical systems. 
% studied the symmetries, constants of the motion and reduction of forced mechanical systems.
% (see Ref.~\cite{de_leon_symmetries_2021}, see also Ref.~\cite{lopez-gordon_geometry_2021}). 
In this paper we have obtained the discrete counterparts of our Noether's theorem for forced mechanical systems, 
% \cite[Theorem 1]{de_leon_symmetries_2021},
and of our reduction lemma.
% \cite[Lemma 15]{de_leon_symmetries_2021}.

In the second paper, %\cite{de_leon_geometric_2022}, 
we developed a Hamilton-Jacobi theory for forced Hamiltonian and Lagrangian systems and characterized the complete solutions, relating them with constants of the motion in involution. We also studied the reduction and reconstruction of solutions of the Hamilton-Jacobi problem for systems with symmetry, as well as the reduction of the Hamilton-Jacobi problem for a \v{C}aplygin system to the Hamilton-Jacobi problem of a forced Lagrangian system.

This paper has left some open questions to be studied elsewhere. To start with, it would be interesting to obtain sufficient conditions for an exact discrete force to be Rayleigh. In fact, it seems that natural Lagrangian subject to continuous Rayleigh forces define an exact discrete force which is Rayleigh (see Conjecture \ref{conjecture_Rayleigh}). It would also be quite useful to obtain expressions for discrete Rayleigh potentials with other variational integrators (as we have done for the midpoint rule). An alternative approach for the construction of a discrete Hamilton-Jacobi theory is the discrete vector field approach \cite{de_leon_geometry_2018,cresson_continuous_2014,cresson_continuous_2015}. Moreover, we could explore if the results we have obtained for forced continuous systems, such as other types of symmetries or complete solutions of the Hamilton-Jacobi problem, have a discrete version. Another possible line of research is the development of a reduction method for forced discrete systems \cite{fernandez_lagrangian_2020,fernandez_lagrangian_2016,jalnapurkar_discrete_2006}, or the study of forced discrete systems on Lie groupoids \cite{marrero_discrete_2006}.
Our Hamilton-Jacobi theory for forced discrete systems could have applications in optimal control, as the conservative one does \cite{lee_optimal_2014}. In addition, one could study the case in which the forced system is a ``small perturbation'' of the conservative system \cite{hairer_invariant_1999}, both in the continuous and discrete versions.

\section*{Acknowledgements}

We are thankful for the suggestions of the referee, which have remarkably improved the clarity of the text.
The authors acknowledge financial support from the Spanish Ministry of Science and Innovation (MICINN), under grants PID2019-106715GB-C21 and ``Severo Ochoa Programme for Centres of Excellence in R\&D'' (CEX2019-000904-S). Manuel Lainz wishes to thank MICINN and the Institute of Mathematical Sciences (ICMAT) for the FPI-Severo Ochoa predoctoral contract PRE2018-083203. Asier López-Gordón would like to thank MICINN and ICMAT for the predoctoral contract PRE2020-093814.

\section*{Data Availability}
The data that support the findings of this study are openly available at the following URL/DOI:
\url{https://github.com/aslogor/discrete_Hamilton_Jacobi_scripts}
% (omitted for anonymity)
% The data that supports the findings of this study are available within the article.
% Data sharing is not applicable to this article as no new data were created or analysed in this study.

\printbibliography
% \nocite{*}
% \nocite{ohsawa_discrete_2011,de_leon_hamilton-jacobi_2014,cresson2014continuous,cresson2015continuous,sato_martin_de_almagro_discrete_2020,ortega_dynamics_2004,birtea_asymptotic_2007,coquinot_general_2020}

\end{document}